\documentclass[article]{IEEEtran}
\usepackage{amsmath}
\usepackage{amsfonts}
\usepackage{amssymb}
\usepackage{amsthm}
\usepackage{tabularx}
\usepackage{mathrsfs} 

\usepackage{url}
\usepackage[colorlinks]{hyperref}

\newtheorem{corollary}{Corollary} 
\newtheorem{proposition}{Proposition}

\usepackage{stfloats}
\usepackage{float}
\usepackage{graphicx}
\usepackage{xcolor}
\usepackage{subfigure}
\usepackage{array}
\makeatletter
\newcommand{\thickhline}{%
    \noalign {\ifnum 0=`}\fi \hrule height 1pt
    \futurelet \reserved@a \@xhline
}
\newcolumntype{"}{@{\hskip\tabcolsep\vrule width 1pt\hskip\tabcolsep}}
\makeatother

\makeatletter
\def\blfootnote{\xdef\@thefnmark{}\@footnotetext}
\makeatother

\begin{document}

\title{Enormous Fluid Antenna Systems (E-FAS) for Multiuser MIMO: Channel Modeling and Analysis}
	
\author{Farshad~Rostami~Ghadi,~\IEEEmembership{Senior~Member},~\textit{IEEE}, 
            Kai-Kit~Wong,~\IEEEmembership{Fellow},~\textit{IEEE},\\
            Masoud~Kaveh,~\IEEEmembership{Member},~\textit{IEEE}, 
            Wee~Kiat~New,~\IEEEmembership{Member},~\textit{IEEE}, 
            Chan-Byoung Chae,~\IEEEmembership{Fellow},~\textit{IEEE}, 
            and~Lajos~Hanzo,~\IEEEmembership{Life Fellow},~\textit{IEEE}
\vspace{-8mm}
}

\maketitle

\blfootnote{The work of  F. Rostami Ghadi, K. K. Wong, and W. K. New is supported by the Engineering and Physical Sciences Research Council (EPSRC) under Grant EP/W026813/1. The work of C.-B. Chae is supported by the Institute for Information and Communication Technology Planning and Evaluation (IITP)/National Research Foundation of Korea (NRF) grant funded by the Ministry of Science and ICT (MSIT), South Korea, under Grant RS-2024-00428780 and 2022R1A5A1027646. The work of L. Hanzo is supported by the Engineering and Physical Sciences Research Council (EPSRC) projects Platform for Driving Ultimate Connectivity (TITAN) under Grant EP/X04047X/1 and Grant EP/Y037243/1.}

\blfootnote{\noindent F. Rostami Ghadi, K. K. Wong, and W. K. New are with the Department of Electronic and Electrical Engineering, University College London, WC1E 7JE London, United Kingdom. K. K. Wong is also affiliated with Yonsei Frontier Lab., Yonsei University, Seoul, Republic of Korea (e-mail: $\rm \{f.rostamighadi, kai\text{-}kit.wong\}@ucl.ac.uk$, $\rm aven.wknew@yahoo.com$).}
\blfootnote{\noindent M. Kaveh is with the Department of Information and Communication Engineering, Aalto University, Espoo, Finland. (e-mail: $\rm masoud.kaveh@aalto.fi$).}
\blfootnote{\noindent C.-B. Chae is with the School of Integrated Technology, Yonsei University, Seoul, 03722, Republic of Korea (e-mail: $\rm cbchae@yonsei.ac.kr$).}
\blfootnote{\noindent L. Hanzo is with the School of Electronics and Computer Science, University of Southampton, Southampton, United Kingdom (e-mail: $\rm lh@ecs.soton.ac.uk$).}
\blfootnote{\noindent Corresponding Author: Kai-Kit Wong.}

\begin{abstract}
Enormous fluid antenna systems (E-FAS), the system concept that utilizes position reconfigurability in the large scale, have emerged as a new architectural paradigm where intelligent surfaces are repurposed from passive, controllable reflectors into multi-functional electromagnetic (EM) interfaces that can route guided surface waves (SWs) over surfaces such as walls, ceilings, and facades, as well as emit space waves toward target receivers. This expanded functionality introduces a new mode of signal propagation, facilitating novel forms of wireless communication. {\color{blue} In this paper, we investigate this new paradigm and develop an end-to-end statistical channel model for E-FAS-assisted wireless transmission by combining guided-SW propagation with small-scale fading on the space wave segments. Specifically, we establish that the SW-assisted channel is conditionally circularly symmetric complex Gaussian with a random covariance scale, whereas its unconditional distribution is generally a Gaussian  scale mixture. We derive the general equivalent channel covariance in the presence of spatial correlation and identify the channel hardening and spatially white conditions under which the independent and identically distributed (i.i.d.) Gaussian model becomes accurate. Under this tractable specialization, established Rayleigh fading and multiuser multiple-input multiple-output (MU-MIMO) results are applied to characterize single-user outage and ergodic rate performance, as well as the post-zero-forcing (ZF) signal-to-interference plus noise ratio (SINR) and sum-rate. Full cascade simulations validate the statistical characterization and quantify the degradation caused by base station (BS)-side spatial correlation. Numerical evaluations are physically calibrated at $26.4$~GHz using experimentally reported E-FAS parameters, while conventional direct links are modeled using the 3GPP indoor-hotspot (InH) propagation model. The results demonstrate substantial performance gains over an ideal passive reconfigurable intelligent surface (RIS) under moderate and long propagation distances by transferring most of the end-to-end distance to the guided-SW pathway with much less propagation loss.}
\end{abstract}


\vspace{-2mm}
\section{Introduction}
\subsection{Background}
\IEEEPARstart{T}{he ever-growing} demand of network capacity, the more ambitious targets on various key performance indicators (KPIs) and the constant push by emerging services such as augmented or extended reality, holographic communication, and internet-of-things (IoT), etc., continue to put pressure on existing mobile communication technologies, heading towards the sixth generation (6G) \cite{Tariq-2020,wang2023ontheroad}. Current physical layer has been relying on multiuser multiple-input multiple-output (MU-MIMO) to provide the needed capacity \cite{Villalonga2022spectral}. As a matter of fact, since the pioneering work by Foschini \cite{Foschini-1996} and Telatar \cite{Telatar-1999} in the late 90s, MIMO has become a principal technology in the physical layer. MIMO subsequently evolved into MU-MIMO in the fourth generation (4G) \cite{wong2000opt,wong2002per,wong2003jcd,Vishwanath-2003,Spencer-2004}, becoming so powerful that saw it to scale up in the massive version in the fifth generation (5G) \cite{Marzetta-2010,Larsson-2014}, from $8$ antennas in 4G to $64$ antennas in 5G at the base station (BS).

Looking to delivering 6G, there has been strong interest to deploy even more antennas at the BS \cite{10379539} and migrate toward higher carrier frequencies \cite{rap2013mil}. However, the latter increases the vulnerability of wireless links due to blockage, penetration loss, and unfavorable propagation conditions \cite{mol2012wire}. It is worth noting that the MU-MIMO precoding with a larger array will increase the peak-to-average power ratio (PAPR) even more, which greatly degrades the efficiency of power amplifiers and increases the power consumption. With blockage in the high frequency bands, there is a legitimate concern that the power consumption of the BS cannot keep up and the communication performance would be considerably limited, especially when the line-of-sight (LoS) is obstructed, and where conventional space wave transmission suffers from severe path loss due to spherical wave spreading and few diffraction capabilities \cite{heath2016an}.

To tackle this, reconfigurable intelligent surfaces (RIS) \cite{Liaskos-2018,basar2019wireless} have been widely studied recently, as a means of reshaping the wireless propagation environment by programmable phase and amplitude control on the incident electromagnetic (EM) waves for `smart' reflections. Although these architectures can redirect or partially forward energy toward desired locations, their operations are fundamentally based on three-dimensional (3D) space wave propagation. Consequently, the achievable performance gains are still constrained by spherical power decay and, in many situations, by the doubly fading and path loss effects as well. Additionally, the reliance on fine-grained EM control and accurate channel state information (CSI) across large surfaces introduces severe system complexity, particularly in dense multiuser deployments \cite{An-2022}.

While MU-MIMO systems should continue to dominate the physical layer for years to come, it is evident that we need other technologies to complement MU-MIMO. In other words, we should seek to enhance its capability without the need of increasing the number of antennas. 
{\color{blue}
Although RISs can create favorable propagation conditions, their assisted links remain fundamentally limited by cascaded space wave propagation, meaning that the channel gain of a far-field RIS-assisted link contains the product of the BS-to-RIS and RIS-to-user equipment (UE) space wave losses, inducing a severe multiplicative path-loss penalty that can greatly limit the achievable link gain, particularly over extended propagation distances.}


\vspace{-3mm}
\subsection{Reconfigurable Antenna-aided Communications}
One possible technology to elevate MIMO is the adoption of reconfigurable antennas \cite{Bernhard-2007} and the first effort dates back to early 2000s by Cetiner {\em et al.}~in \cite{1367557}. However, fast forward to today, reconfigurable antennas are still yet to revolutionize wireless communication technologies. This stems from a gap between existing reconfigurable antenna capabilities and the capabilities needed to transform wireless communications. To address this, Wong {\em et al.}~introduced the fluid antenna system (FAS) \cite{wong2021fluid,wong2020perfrom} which is a hardware-agnostic system concept that considers the antenna as a reconfigurable physical layer resource to broaden system design and network optimization, and inspire next-generation reconfigurable antennas \cite{new2025tut,Lu-2025,hong2025contemporary,new2025flar,wu2024flu}. FAS represents all forms of reconfigurable antenna-aided wireless communications systems. In practice, FAS may be implemented by movable elements \cite{zhu2024historical}, liquid-based antennas \cite{shen2024design,Shamim-2025}, metamaterial-based antennas \cite{Zhang-jsac2026,Liu-2025arxiv}, reconfigurable pixels \cite{zhang2024pixel,tong-2025pixel,Wong-wc2026,qiao-2026,zhang-2026pixel}, etc. In \cite{tong2025designs}, Tong {\em et al.}~compared various FAS prototypes.


After the seminal work \cite{wong2021fluid,wong2020perfrom}, studies have investigated the benefits of antenna position reconfigurability. These include analytical treatment of spatial correlation \cite{kham2023new}, copula-based analysis under arbitrary fading \cite{rostami2023copula}, Rayleigh fading diversity characterization \cite{new2024fluida}, spatial block correlation modeling \cite{espin2024new}, and finite-blocklength analysis \cite{Zhang-2026fas}. FAS has also been extended to MIMO-FAS \cite{G10_new2024MIMO-FAS}, continuous FAS \cite{psomas2023con}, and MU-MIMO systems \cite{xu2025capacity,cheng2024sum}. Also, CSI acquisition has been studied using channel sparsity \cite{xu2024channel}, learning-based methods \cite{zhang2024learning}, Nyquist sampling with maximum likelihood estimation \cite{new2025channel}, and successive Bayesian reconstruction \cite{zhang2023successive}.


FAS has also been integrated with RIS, either by equipping the mobile receiver with FAS \cite{ghadi2024on} or by enabling position reconfigurability across RIS elements, giving rise to fluid RIS (FRIS) \cite{salem2025first,xiao2025fluid,ghadi2025perfrom,Xiao-twc2026} and also fluid integrated reflecting and emitting surfaces (FIRES) \cite{ghadi2025fires}. These metasurface-based extensions dynamically reconfigure reflecting and radiating elements, but still rely on 3D space wave propagation and complex beamforming, which limit their power efficiency.

To fundamentally address the limitations of wireless communication, \cite{wong2021vision} proposed a new paradigm in which intelligent surfaces are leveraged not only as emitters, but as engineered propagation media. In this paradigm, signal behavior is largely dictated by surface geometry, and interference management becomes much more predictable, while the guided nature of propagation enables highly reliable links that closely resemble persistent LoS conditions \cite{liu2024path,chu2024on,chu2025on}. {\color{blue}This approach was technically discussed in \cite{Shojaeifard-2022}, where reconfigurable surfaces were used as intelligent propagation media supporting guided surface-wave (SW) transmission and reconfigurable propagation pathways. This paradigm was recently referred to as an enormous FAS (E-FAS) in \cite{wong2025efas}. The term ``enormous'' does not refer to a conventional FAS with a dense set of candidate antenna ports over an enlarged aperture of a terminal. Rather, it reflects the realization of large-scale spatial reconfigurability over distributed intelligent surfaces in an environment.}\footnote{\color{blue}Thus, the strong port-to-port spatial correlation typically associated with compact FAS implementations is not an intrinsic characteristic of E-FAS.}

In E-FAS, intelligent surfaces operate as guided EM interfaces that convert incident space waves into SWs and vice versa. The resulting SWs propagate along engineered surfaces through confined and reconfigurable pathways, thereby enabling substantially lower propagation loss than conventional 3D space-wave transmission. {\color{blue}By routing energy along walls, ceilings, or building exteriors and radiating wirelessly only in the first and final hops, E-FAS can re-establish LoS connectivity in shadowed environments where even RIS cannot, and reduce the required BS radiated power. Recent experimental results further demonstrate that reconfigurable surface regions can directly form high-gain free-space beams while remaining integrated with the guided SW propagation structure~\cite{LiuEFAS2026}.}\footnote{\color{blue}Like other programmable surfaces such as RIS, the reconfigurable surfaces in E-FAS may require tuning and control circuitry for reconfiguration. However, it should be noted that any transition between SW and space wave does not require an active RF amplification stage at the surface \cite{LiuEFAS2026}.} 

{\color{blue}Within the broader family of spatially reconfigurable architectures, it is useful to distinguish E-FAS from the emerging movable antenna \cite{zhu2024model} and pinching antenna systems \cite{ntt,ding2025flex}. Movable antennas, as a particular implementation under the general framework of FAS, reposition antenna elements within a local region to exploit favorable channel variations at BS or UE. The degree of position reconfigurability and the extent to which the propagation environment can be engineered are fundamentally different from those provided by E-FAS. For this reason, movable antennas (and also FAS more generally) and E-FAS should be regarded as complementary rather than mutually exclusive technologies.\footnote{\color{blue}RIS is more directly comparable to E-FAS, as both rely on intelligent surfaces deployed within the environment to enhance communication performance. Nevertheless, RIS and E-FAS operate according to distinct propagation principles, despite ultimately pursuing a similar objective.} Pinching antennas \cite{ding2025flex}, on the other hand, have been proposed to create reconfigurable radiating points along linear dielectric waveguides, thereby enabling radiation-position reconfigurability within the environment. At a conceptual level, research on ``pinching antennas'' shares important similarities with E-FAS. Nevertheless, two key distinctions deserve recognition. First, the vision of large-scale radiation-point reconfiguration was introduced in 2021 in \cite{wong2021vision}, more than one year before the first pinching-antenna paper by NTT DOCOMO \cite{ntt}, and well before the subsequent system-level studies published by the communications community, e.g., \cite{ding2025flex}. As such, the underlying principle adopted in \cite{ding2025flex} and related subsequent works should be attributed to \cite{wong2021vision}, which first introduced the E-FAS concept.\footnote{\color{blue}Moreover, E-FAS envisages the use of 3D intelligent surfaces and therefore represents a more general framework than the linear waveguide configurations considered in the emerging pinching-antenna literature.}

More importantly, the E-FAS framework considered in this paper is grounded in experimentally demonstrated reconfigurable waveguide technologies, as evidenced by the prototypes and measurements reported in \cite{chu2024on,LiuEFAS2026}.\footnote{\color{blue}The same level of hardware validation is not yet available for the system models considered in \cite{ding2025flex} and related pinching-antenna studies.} In fact, the pinching-antenna prototype demonstrated by NTT DOCOMO was not designed to provide dynamically reconfigurable pinch points. Therefore, the movable pinch points assumed in subsequent system-level studies have not been experimentally realized. Moreover, important practical characteristics, including the transition loss from the waveguide mode to free-space radiation and the propagation loss along the waveguide, remain insufficiently characterized. In contrast, the waveguide architecture underpinning E-FAS is inherently reconfigurable and supports beam focusing along the surface, thereby significantly reducing propagation loss over the guided pathway. These capabilities have already been experimentally validated by the hardware prototypes reported in \cite{chu2024on,LiuEFAS2026}.
}

\vspace{-2mm}
\subsection{Motivation and Contributions}
Despite the propagation advantages envisaged by E-FAS, its analytical performance is not well understood. There is no analytical characterization of the effective end-to-end channel statistics induced by E-FAS, nor of its KPIs such as outage probability and ergodic capacity. A fundamental understanding of how SW-assisted propagation translates into classical communication-theoretic performance metrics is missing. 

{\color{blue}Different from \cite{wong2025efas}, this paper provides the first closed-form, end-to-end statistical characterization of E-FAS-assisted channels and establishes a tractable framework for evaluating their communication-theoretic performance. The principal contribution is the development and statistical characterization of the underlying E-FAS channel model. Once the conditions leading to a Gaussian or approximately Gaussian equivalent channel are established, classical fading channel and MU-MIMO results are applied to characterize outage probability, ergodic rate, and post-zeroforcing (ZF) performance. These metrics are employed to quantify the communication-theoretic implications of the proposed channel model and its impact on system-level performance. The resulting framework enables rigorous benchmarking, scalability analysis, and analytical insights that are not readily accessible through EM simulations alone.} The distinctive contributions of our work, compared to the literature, are highlighted in Table \ref{tab:contrast}. Technically speaking, our contributions are summarized as follows:

\begin{itemize}
{\color{blue}
\item
We develop a statistical end-to-end channel model that connects the physical BS-to-surface excitation, guided SW propagation, reconfigurable surface re-radiation, and final space-wave propagation. We then establish the exact conditional Gaussian  characterization of the SW-assisted channel and reveal that its unconditional distribution is generally compound Gaussian as the  conditional covariance depends on the random radiation-to-user channel.  We further derive the general equivalent channel covariance in the  presence of spatial correlation and identify the conditions where the channel reduces to an approximately spatially white complex Gaussian model. The accuracy of this approximation is validated using  samples generated from the underlying cascaded channel.

\item 
Building on the channel model, we employ established fading channel results to characterize the outage probability and ergodic capacity under both a-no-CSI at the transmitter (CSIT) transmission benchmark and perfect CSIT maximum ratio transmission (MRT). The resulting expressions show how the SW-assisted channel gain affects the outage coding gain and high signal-to-noise ratio (SNR) capacity offset, while MRT additionally provides the conventional transmit array and diversity gains.

%

\item We extend the nominal guided-SW channel model to account for practical surface and propagation imperfections. We introduce a structured representation of impedance-induced amplitude and phase perturbations and model reflections, bends, junctions, and routing discontinuities through sparse secondary guided components. This formulation distinguishes the calibrated nominal E-FAS route from its perturbed practical realization.

\item We characterize the multiuser statistical dependence induced by shared E-FAS propagation. We derive the conditional cross-covariance between different users' equivalent channels and identify route overlap, common BS-to-surface propagation, and spatial correlation as sources of inter-user coupling. This analysis establishes the conditions under which independent user channels and classical Rayleigh MU-MIMO models constitute accurate results.

\item Finally, we develop a physically calibrated evaluation framework using experimentally reported E-FAS parameters at $26.4$~GHz. The proposed channel model is evaluated under an explicit segmented propagation geometry and benchmarked against 3GPP direct transmission and an ideal passive RIS under a common deployment setting. The results reveal superior performance of guided-SW propagation, greatly outperforming conventional space wave and passive RIS-assisted transmission.}
\end{itemize}


\vspace{-3mm}
\subsection{Organization and Notations}
The remainder of this paper is organized as follows. Section~\ref{sec:system_model} introduces the system and channel model. Section~\ref{sec:eq_channel} derives the equivalent end-to-end channel distribution of E-FAS-assisted propagation. Then Sections~\ref{sec:single_user} and~\ref{sec:multi-user} analyze the single-user and MU transmission scenarios, respectively. Our numerical results are presented in Section~\ref{sec:numerical}, and concluding remarks are given in Section~\ref{sec:conclusion}.

{\em Notations}---Scalars, vectors, and matrices are denoted by lowercase letters, bold lowercase letters, and bold uppercase letters, respectively. The operators $(\cdot)^T$, $(\cdot)^H$, and $\mathrm{tr}(\cdot)$ denote transpose, Hermitian transpose, and trace. The expectation operator is denoted by $\mathbb{E}\{\cdot\}$. The notation $\mathcal{CN}(\boldsymbol{\mu},\mathbf{C})$ represents a circularly symmetric complex Gaussian distribution with mean $\boldsymbol{\mu}$ and covariance matrix $\mathbf{C}$. The identity matrix of appropriate dimension is denoted by $\mathbf{I}$. The Euclidean norm is denoted by $\|\cdot\|$. Unless otherwise stated, $\log(\cdot)$ and $\log_2(\cdot)$ denote the natural and base-$2$ logarithms, respectively. The main symbols and variables used throughout are summarized in Table \ref{tab:not}.

\begin{table*}[t]
\centering
\caption{\color{blue}Comparison of the modeling scope and analytical coverage with the existing literature.}\label{tab:contrast}
\renewcommand{\arraystretch}{1.12}
\setlength{\tabcolsep}{3.5pt}
\scriptsize
\resizebox{.95\linewidth}{!}{
\begin{tabular}{p{9cm}||c|c|c|c|c|c|c|c|c}
		\thickhline
		\textbf{Features}
		& \textbf{\cite{wong2025efas}}
		& \textbf{\cite{wong2021vision}}
		& \textbf{\cite{liu2024path}}
		& \textbf{\cite{chu2024on}}
		& \textbf{\cite{chu2025on}}
		& \textbf{\cite{salem2025first}}
		& \textbf{\cite{ghadi2025perfrom}}
		& \textbf{\cite{ghadi2025fires}}
		& \textbf{This paper} \\
		\thickhline
		
		E-FAS
		& $\checkmark$ & & & & & & & & $\checkmark$ \\
		\hline
		
		Distance-dependent guided SW propagation
		& $\checkmark$ & $\checkmark$ & $\checkmark$ & $\checkmark$
		& $\checkmark$ & & & & $\checkmark$ \\
		\hline
		
		Environment-embedded reconfigurable surfaces
		& $\checkmark$ & $\checkmark$ & $\checkmark$ & $\checkmark$
		& $\checkmark$ & $\checkmark$ & $\checkmark$ & $\checkmark$
		& $\checkmark$ \\
		\hline
		
		Geometrically reconfigurable radiating aperture
		& $\checkmark$ & $\checkmark$ & $\checkmark$ & $\checkmark$
		& $\checkmark$ & $\checkmark$ & $\checkmark$ & $\checkmark$
		& $\checkmark$ \\
		\hline
		
		Statistical characterization of a cascaded equivalent-channel gain
		& & & & & & & $\checkmark$ & & $\checkmark$ \\
		\hline
		
		Spatial correlation in a surface-assisted channel model
		& & & & & & $\checkmark$ & $\checkmark$
		& $\checkmark$ & $\checkmark$ \\
		\hline
		
		
		Exact conditional Gaussian characterization of the
		E-FAS cascaded channel
		& & & & & & & & & $\checkmark$ \\
		\hline
		
		Unconditional compound-Gaussian characterization of the
		E-FAS cascaded channel
		& & & & & & & & & $\checkmark$ \\
		\hline
		
		Conditions for the spatially white Gaussian E-FAS special case
		& & & & & & & & & $\checkmark$ \\
		\hline
		
	Impact of BS-side spatial correlation on the post-ZF SINR distribution
		& & & & & & & & & $\checkmark$ \\
		\hline
		
		Outage probability characterization
		& & & & & & & $\checkmark$ & & $\checkmark$ \\
		\hline
		
		Ergodic rate or sum-rate analysis
		& & & & & & $\checkmark$ & $\checkmark$
		& $\checkmark$ & $\checkmark$ \\
		\thickhline
	\end{tabular}}
\vspace{-2mm}
\end{table*}

\begin{table}[]	
\centering
\caption{Summary of Main Notations}\label{tab:not}
\resizebox{.9\columnwidth}{!}{
	\begin{tabular}{ll}
		\hline
		\textbf{Symbol} & \textbf{Description} \\
		\hline
		$M$ & Number of antennas at the BS \\
		$K$ & Number of UEs \\
		\color{blue}$N_s$
		& 	\color{blue}Number of effective surface excitation ports \\
			\color{blue}$N_L$
			\color{blue}& 	\color{blue}Number of effective surface radiation modes \\
		$\mathbf{x}$ & Transmit signal vector at the BS \\
		$x_u$ & Information symbol intended for UE $u$ \\
		$P$ & Total BS transmit power \\
		$P_u$ & Transmit power allocated to UE $u$ \\
		$\mathbf{w}_u$ & Precoding vector for UE $u$ \\
		{\color{blue}$d_\mathrm{BS}$} & \color{blue}{BS-to-surface excitation distance} \\
				{\color{black}$d_{\mathrm{sw},u}$} & \color{blue}{Guided-SW propagation distance for UE $u$} \\
			{\color{blue}$d_{\mathrm{RU},u}$} & {\color{blue}Surface radiation region-to-UE distance}\\
		{\color{blue} $d_{\mathrm{DL},u}$} & {\color{blue}Direct BS-to-UE distance}\\
		{\color{blue} $D$} & {\color{blue} Segmented end-to-end distance}\\
		$\mathbf{H}_{\mathrm{BS\textnormal{-}sur}}$ & BS-to-surface space wave channel matrix \\
		\color{blue}$\mathbf{R}_{\mathrm{s}}$
			& \color{blue}Spatial correlation matrix across the surface excitation ports \\
			\color{blue}$\mathbf{R}_{\mathrm{B}}$
			& \color{blue}Spatial correlation matrix across the BS antennas \\
			\color{blue}$\mathbf{R}_{\mathrm{L},u}$
			& \color{blue} Spatial correlation matrix across the effective radiation modes \\
			\color{blue}$\mathbf{R}_{\mathrm{DL},u}$
			& \color{blue}BS-side spatial correlation matrix of the direct link to UE $u$ \\
		$\mathbf{H}_{\mathrm{sur}}$ & SW propagation channel  \\
		$\color{blue}\overline{\mathbf{H}}_{\mathrm{sur}}$
		& \color{blue}Calibrated nominal SW transfer matrix \\
	$\Delta\mathbf{H}_{\mathrm{imp}}$
	& \color{blue}SW perturbation induced by impedance variations \\
	$\color{blue}\mathbf{H}_{\mathrm{sec}}$
	& \color{blue}Aggregate secondary guided SW components \\
	$\color{blue}\Delta\mathbf{H}_{\mathrm{sur}}$
	& \color{blue}Total SW deviation,
	$\color{blue}\Delta\mathbf{H}_{\mathrm{imp}}+\mathbf{H}_{\mathrm{sec}}$ \\
		\textcolor{blue}{$\mathbf W_{\rm rad}$} & \textcolor{blue}{Reconfigurable surface-radiation mapping matrix}\\
		{\color{blue}$\mathbf h_{\rm rad-UE,u}$} & {\color{blue}Radiation-to-UE channel vector}\\
		$\mathbf{h}_u^{\mathrm{dl}}$ & Direct BS-to-UE channel vector \\
		$\mathbf{h}_u^{\mathrm{eq}}$ & Equivalent end-to-end BS-to-UE channel vector \\
		$h_u^{\mathrm{eq}}$ & Effective channel for UE $u$ \\
		$\beta_{DL}$ & Large-scale gain of the direct BS-to-UE link \\
		$\Omega_{\mathrm{sw}}$ & SW-assisted channel gain \\
		$\Omega_{\mathrm{eq}}$ & Equivalent end-to-end channel variance \\
		$\sigma^2$ & Thermal noise variance at the UE \\
		$\gamma_u$ & Instantaneous received SINR at UE $u$ \\
		$\gamma_u^{\mathrm{ZF}}$ & Post-ZF SINR for UE $u$ \\
		$\rho$ & Transmit SNR parameter \\
		$\mathbf{H}_{\mathrm{eq}}$ & Equivalent BS-to-UE channel matrix \\
		$\mathbf{G}$ & Gram matrix $\mathbf{H}_{\mathrm{eq}}^H \mathbf{H}_{\mathrm{eq}}$ \\
		$m$ & ZF diversity parameter \\
		$R_0$ & Target transmission rate  \\
		$P_{\mathrm{out}}$ & Outage probability \\
		$C$ & Ergodic capacity \\
		$R_{\mathrm{sum}}^{\mathrm{ZF}}$ & ZF ergodic sum-rate  \\
		\hline
\end{tabular}}
\vspace{-3mm}
\end{table}

\vspace{-2mm}
\section{System and Channel Model}\label{sec:system_model}
Consider the downlink system as shown in Fig.~\ref{fig:efas_model}(a), where an $M$-antenna BS communicates with $K$ single-antenna UEs in an environment assisted by an E-FAS. The E-FAS comprises distributed metasurface tiles mounted on surrounding structures. 
{\color{blue}These tiles convert incident space waves into guided SWs that propagate along engineered pathways. At selected locations, the surface configuration is modified to re-radiate the guided field as a directive space wave toward the intended UEs~\cite{LiuEFAS2026}.}
This structure yields a composite channel with (i) a BS-to-surface excitation segment; 
(ii) {\color{blue}a calibrated nominal SW propagation stage;
(iii) a reconfigurable surface re-radiation stage;} and 
(iv) a final space wave hop to each UE.

\vspace{-2mm}
\subsection{Transmit Signal Model}
For the model considered, let $x_u\in\mathbb{C}$ denote the data symbol intended for UE~$u$, normalized such that $\mathbb{E}\{|x_u|^2\}=1$. The BS applies a linear precoding $\mathbf{w}_u\in\mathbb{C}^{M}$ to this symbol, resulting in the transmit vector
\begin{align}
\mathbf{x}= \sum_{u=1}^{K} \sqrt{P_u}\,\mathbf{w}_u x_u ,
\end{align}
where $P_u$ is the power allocated to UE~$u$ and $\sum_{u=1}^{K} P_u = P$ is the power budget of the BS. Unless otherwise stated, $P_u=P/K, \forall u$, is adopted for analytical clarity.

\vspace{-2mm}
\subsection{BS-to-Surface Excitation}
%

{\color{blue}
The channel between the BS array and the $N_s$ effective excitation ports is denoted by $\mathbf{H}_{\mathrm{BS\textnormal{-}sur}}\in\mathbb{C}^{N_s\times M}$. To account for spatial correlation at both the BS array and the surface excitation ports, we adopt the Kronecker model
\begin{align}
\mathbf{H}_{\mathrm{BS\textnormal{-}sur}}=\sqrt{\beta_{\mathrm{BS}}}\,\mathbf{R}_{\mathrm{s}}^{1/2}\mathbf{G}_{\mathrm{BS\textnormal{-}sur}}\mathbf{R}_{\mathrm{B}}^{1/2},\label{eq:h-corr}
\end{align}
where $\mathbf{G}_{\mathrm{BS\textnormal{-}sur}}\in\mathbb{C}^{N_s\times M}$ contains independent and identically distributed (i.i.d.) $\mathcal{CN}(0,1)$ entries. Moreover, $\mathbf{R}_{\mathrm{s}}\in\mathbb{C}^{N_s\times N_s}$ and $\mathbf{R}_{\mathrm{B}}\in\mathbb{C}^{M\times M}$ are Hermitian positive semidefinite spatial correlation matrices associated with the surface excitation ports and the BS antennas, respectively. They are assumed to have unit diagonal entries, such that $\operatorname{tr}(\mathbf{R}_{\mathrm{s}})=N_s$ and $\operatorname{tr}(\mathbf{R}_{\mathrm{B}})=M$, thereby preserving $\beta_{\mathrm{BS}} = \beta_{0}\left(\frac{d_{\rm BS}}{d_0}\right)^{-n_{\rm BS}},$ as the average large-scale gain of this propagation segment, where $d_{\rm BS}$ is the BS-to-surface excitation distance, $d_0$ is the reference distance and $n_{\rm BS}$ is the space wave path-loss exponent. The spatially white channel model is the special case with $\mathbf{R}_{\mathrm{s}}=\mathbf{I}_{N_s}$ and $\mathbf{R}_{\mathrm{B}}=\mathbf{I}_M$. These correlation matrices are introduced to capture general correlation effects arising from the array geometry, aperture configuration, and propagation environment. They do not represent the spatial correlation among dense antenna ports that is typically associated with conventional compact FAS implementations.}

\begin{figure}[t]
\centering
\subfigure[General E-FAS system model and distance parameters.]
{\includegraphics[width=.4\textwidth]{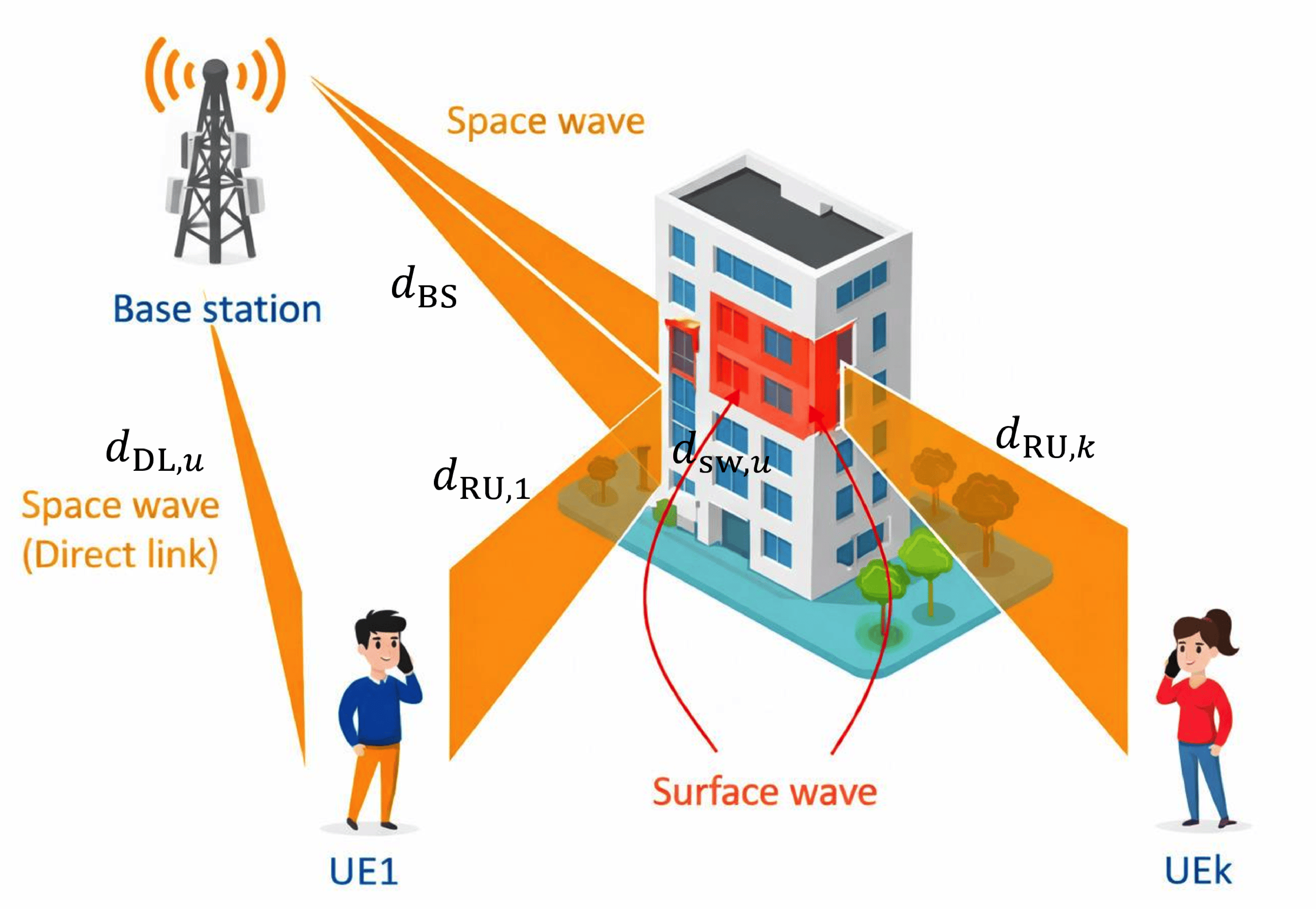}\label{fig:efas_model_a}}\\
\subfigure[Trapezium geometry adopted for the numerical evaluation.]
{\includegraphics[width=.4\textwidth]{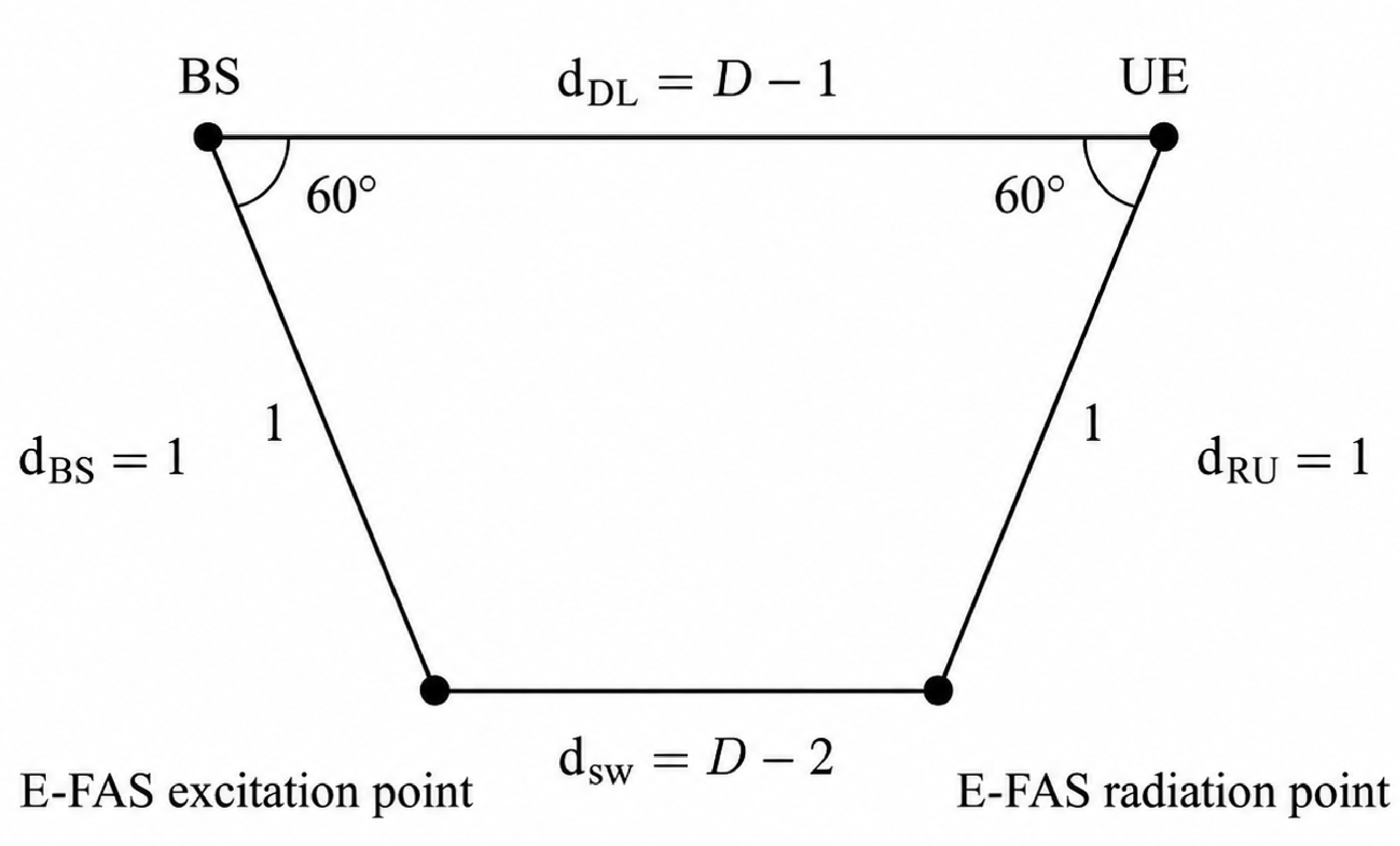}\label{fig:efas_model_b}}
\caption{{\color{blue}(a) A downlink E-FAS-aided environment with an $M$-antenna BS communicating with $K$ single-antenna UEs. The E-FAS-assisted route to UE $u$ consists of the BS-to-surface distance $d_{\rm BS}$, the guided-SW distance $d_{{\rm sw},u}$, and the surface re-radiation-to-UE distance $d_{{\rm RU},u}$, while $d_{{\rm DL},u}$ is the direct BS-to-UE distance. (b) The trapezium geometry adopted for the numerical evaluation, where $d_{\rm BS}=d_{\rm RU}=1$~m, $d_{\rm sw}=D-2$, and $d_{\rm DL}=D-1$.}}\label{fig:efas_model}
\vspace{-4mm}
\end{figure}

\vspace{-2mm}
{\color{blue}\subsection{Nominal and Perturbed SW Propagation}
Once excited,\footnote{This is done by ensuring that the E-FAS has a Brewster's angle matching the incident angle of the radio wave from the BS \cite{liu2024path}.} the EM field propagates along the reconfigurable surface as a guided SW. For the configured route to UE $u$, let $d_{{\rm sw},u}$ represent the guided propagation distance along the engineered surface pathway. For a reference distance $d_0$, its nominal envelope is modeled as
$
H_{\rm sw}(d_{{\rm sw},u})
=
A_0
\left(\frac{d_{{\rm sw},u}}{d_0}\right)^{-n_{\rm sw}/2}
10^{-\kappa_{\rm sw}(d_{{\rm sw},u}-d_0)/20}
e^{-j\beta d_{{\rm sw},u}}
$,  where $A_0$ incorporates the excitation and coupling efficiencies at the reference distance, $n_{\rm sw}$ denotes the effective power domain distance exponent of the configured guided SW pathway, $\kappa_{\rm sw}$ denotes an additional material/leakage attenuation coefficient in dB/m, and $\beta$ is the SW phase constant. Importantly, $n_{\rm sw}$ depends on the degree of field confinement provided by the reconfigurable pathway and need not follow the {\em cylindrical} $d^{-1}$ power law of an unconstrained surface wave. Accordingly, $n_{\rm sw}$ is retained as a general confinement parameter. When the distance guided-path loss is characterized directly through a measured per-unit length attenuation coefficient, the additional log-distance spreading term can be omitted by setting $n_{\rm sw}=0$, thereby avoiding double counting of the distance loss.
	
The SW propagation constant $\gamma$ is related to the surface impedance $Z_{\mathrm{sur}}$ according to
	\begin{align}
		\gamma
		=
		\sqrt{
			-\omega^{2}\mu_{0}\epsilon_{0}
			-
			\left(
			-j\omega\epsilon_{0}Z_{\mathrm{sur}}
			\right)^{2}
		},\label{eq:cons}
	\end{align}
where $\omega$ is the operating angular frequency and $(\mu_{0},\epsilon_{0})$ are the free-space permeability and permittivity, respectively \cite{wong2025efas}. For a fixed route and calibrated operating state, the nominal SW transfer matrix is modeled as
\begin{align}
\overline{\mathbf{H}}_{\mathrm{sur},u}=H_{\mathrm{sw}}(d_{{\rm sw},u})\mathbf{G}_{\mathrm{path},u},
\end{align}
where $\mathbf{G}_{\mathrm{path},u}\in\mathbb{C}^{N_s\times N_s}$ specifies the activated tiles, junctions, and interconnections forming the guided route to UE $u$. As a result, the matrix $\overline{\mathbf{H}}_{\mathrm{sur},u}$ is quasi-static for a fixed geometric configuration, operating frequency, and calibrated surface state. This nominal model provides a first-order representation of the guided propagation stage; it does not imply that a practical E-FAS surface is deterministic.
	

{\color{blue}In practice, material inhomogeneities, fabrication tolerances, surface roughness, temperature and humidity variations, mechanical vibrations, and environmental loading may produce spatially and temporally varying deviations from the nominal surface impedance. This behavior may be represented as
\begin{align}
Z_{\mathrm{sur}}(\mathbf{r},t)=\overline{Z}_{\mathrm{sur}}(\mathbf{r})+\Delta Z_{\mathrm{sur}}(\mathbf{r},t),
\end{align}
in which $\overline{Z}_{\mathrm{sur}}(\mathbf{r})$ denotes the calibrated nominal impedance profile and $\Delta Z_{\mathrm{sur}}(\mathbf{r},t)$ collects the corresponding local perturbations. Through the propagation constant in \eqref{eq:cons}, these impedance variations induce amplitude and phase deviations in the SW transfer coefficients. In addition, wall boundaries, corners, bends, junctions, and routing discontinuities may generate secondary guided components that propagate along paths different from the nominal route. Accordingly, the actual SW transfer matrix is represented as
\begin{align}
\mathbf{H}_{\mathrm{sur}}=\overline{\mathbf{H}}_{\mathrm{sur}}+\Delta\mathbf{H}_{\mathrm{imp}}+\mathbf{H}_{\mathrm{sec}}=\overline{\mathbf{H}}_{\mathrm{sur}}+\Delta\mathbf{H}_{\mathrm{sur}},\label{eq:perturbed_surface_channel}
\end{align}
where $\Delta\mathbf{H}_{\mathrm{imp}}$ represents the transfer matrix deviation caused by local impedance variations, surface roughness, contamination, and temperature  environmental fluctuations, while $\mathbf{H}_{\mathrm{sec}}$ contains the aggregate secondary SW components. Thus,
\begin{align}
\Delta\mathbf{H}_{\mathrm{sur}}\triangleq\Delta\mathbf{H}_{\mathrm{imp}}+\mathbf{H}_{\mathrm{sec}}.
\end{align}
	
A structured representation of the impedance induced perturbation at the coefficient level is written as
\begin{align}
\left[\overline{\mathbf{H}}_{\mathrm{sur}}+\Delta\mathbf{H}_{\mathrm{imp}}\right]_{i,j}=\left[\overline{\mathbf{H}}_{\mathrm{sur}}\right]_{i,j}\left(1+\epsilon_{i,j}^{\mathrm{a}}\right)e^{j\epsilon_{i,j}^{\phi}},
\end{align}
where $\epsilon_{i,j}^{\mathrm{a}}$ and $\epsilon_{i,j}^{\phi}$ denote the amplitude and phase errors, respectively.

The secondary guided components may be represented using the sparse multipath expansion
	\begin{align}
		\mathbf H_{\rm sec}
		=
		\sum_{\ell=1}^{L_{\rm sec}}
		a_\ell
		\left(\frac{d_\ell}{d_0}\right)^{-n_{{\rm sw},\ell}/2}
		10^{-\kappa_{{\rm sw},\ell}(d_\ell-d_0)/20}
		e^{-j\beta_\ell d_\ell}
		\mathbf G_\ell, \label{eq:secondary_surface_components}
\end{align}
where $L_{\mathrm{sec}}$ is the number of resolvable secondary SW paths, $a_{\ell}$ incorporates the reflection, scattering, or junction coupling coefficient of the $\ell$th component, $d_{\ell}$ denotes its propagation distance, and $\mathbf{G}_{\ell}\in\mathbb{C}^{N_s\times N_s}$ describes its routing or scattering topology. This representation captures sparse secondary guided components rather than assuming rich conventional space wave multipath. Also, no specific distribution is required for $\Delta Z_{\mathrm{sur}}$, $\epsilon_{i,j}^{\mathrm{a}}$, $\epsilon_{i,j}^{\phi}$, or $a_{\ell}$ in the general channel model.   The calibrated nominal analytical model is recovered by setting $\Delta\mathbf{H}_{\mathrm{sur}}=\mathbf{0}$, equivalently, $\Delta\mathbf{H}_{\mathrm{imp}}=\mathbf{0}$ and  $\mathbf{H}_{\mathrm{sec}}=\mathbf{0}$.
	
The subsequent closed-form analysis adopts the calibrated nominal SW matrix to retain analytical transparency. For any fixed realization of the impedance perturbations and secondary guided components, the conditional channel characterization remains valid after replacing $\overline{\mathbf{H}}_{\mathrm{sur}}$ with the corresponding realization of $\mathbf{H}_{\mathrm{sur}}$. If these quantities vary randomly across transmission intervals, the unconditional channel statistics require additional averaging over $\Delta\mathbf{H}_{\mathrm{imp}}$ and $\mathbf{H}_{\mathrm{sec}}$ and generally do not retain the nominal closed-form expressions. Hence, the deterministic formulation should be interpreted as a calibrated first-order benchmark rather than a complete measurement  model of all practical SW phenomena.}}

\vspace{-2mm}
{\color{blue}\subsection{Reconfigurable Surface Re-Radiation}
At the termination region of a guided route, the local E-FAS surface state is reconfigured so that the guided SW is converted into a directive space wave toward the intended UE. This re-radiation functionality is realized by the EM configuration of the surface, requiring no separate active RF amplification~\cite{LiuEFAS2026}.

We represent the corresponding surface re-radiation mapping for UE $u$ by
\begin{align}
\mathbf W_{{\rm rad},u}\mathbf W_{{\rm rad},u}^{H}=\mathbf I_{N_L},
\end{align}
where $\mathbf W_{\mathrm{rad},u}\in\mathbb C^{N_L\times N_s}$, with $N_L\leq N_s$, specifies the activated radiation modes and their spatial configuration. The above assumption provides a convenient representation of the passive re-radiation mapping and does not correspond to active RF power amplification. Consequently, no amplifier generated noise term is associated with $\mathbf W_{{\rm rad},u}$. Any electrical power required for surface tuning and control is implementation dependent and not included in the link-level propagation and rate model considered in this work. The directional radiation gain and the guided-to-radiated conversion efficiency are accounted for separately in the final space wave channel.}

\vspace{-2mm}
{\color{blue}
\subsection{Re-Radiation-to-UE and Direct BS-to-UE Channels}

The final space wave segment from the configured surface radiation region to UE $u$ is represented by
\begin{align}
	\mathbf h_{{\rm rad-UE},u}
	=
	\sqrt{\eta_{L,u}G_{{\rm rad},u}\ell_{{\rm RU},u}}\,
	\mathbf g_{{\rm rad-UE},u}
	\mathbf R_{L,u}^{1/2},
\end{align}
where $\mathbf g_{{\rm rad-UE},u}\in\mathbb C^{1\times N_L}$ contains i.i.d.~$\mathcal{CN}(0,1)$ entries. The dimensionless relative distance factor of the final space wave segment is defined as $\ell_{{\rm RU},u}=\left(\frac{d_{{\rm RU},u}}{d_0}\right)^{-n_{\rm RU}}$, where $d_{{\rm RU},u}$ denotes the distance between the configured surface radiation region and UE $u$, and $n_{\rm RU}$ is the corresponding space wave path-loss exponent. Importantly, $\ell_{{\rm RU},u}$ captures only the relative distance dependence of the final space wave segment and does not introduce an additional absolute free-space reference gain. The overall E-FAS-assisted route therefore contains only the single absolute free-space reference gain introduced at the BS-to-surface excitation stage. Moreover, $0<\eta_{L,u}\leq 1$ denotes the guided-to-radiated power conversion efficiency, while $G_{{\rm rad},u}>0$ denotes the directional radiation gain of the configured surface aperture. Importantly, $G_{{\rm rad},u}$ is an antenna/directivity gain resulting from coherent radiation and does not represent active RF power amplification. The ideal guided-to-space conversion is recovered by setting $\eta_{L,u}=1$. Also, $\mathbf R_{L,u}\in\mathbb C^{N_L\times N_L}$ is the Hermitian positive semidefinite spatial correlation matrix across the effective radiation modes. This matrix captures correlation caused by the finite radiating aperture, common scattering clusters, and mutual coupling. It is also normalized to have unit diagonal entries, such that $\operatorname{tr}(\mathbf R_{L,u})=N_L$. Consequently,
\begin{align}
	\mathbb E
	\left\{
	\mathbf h_{{\rm rad-UE},u}^{H}
	\mathbf h_{{\rm rad-UE},u}
	\right\}
	=
	\eta_{L,u}G_{{\rm rad},u}
	\ell_{{\rm RU},u}
	\mathbf R_{L,u}.
\end{align}
The spatially white re-radiation-to-UE channel is recovered by setting $\mathbf{R}_{\mathrm{L},u}=\mathbf{I}_{N_L}$. In addition to the SW-assisted route, the BS may have a residual direct propagation path to UE $u$ through penetration, diffraction, or partial blockage. The corresponding direct channel is modeled as
\begin{equation}
	\mathbf{h}_{u}^{\mathrm{dl}}
	=
	\sqrt{\beta_{\mathrm{DL},u}}\,
	\mathbf{g}_{u}^{\mathrm{dl}}
	\mathbf{R}_{\mathrm{DL},u}^{1/2},
\end{equation}
where $\mathbf{g}_{u}^{\mathrm{dl}}\in\mathbb{C}^{1\times M}$ contains i.i.d.~$\mathcal{CN}(0,1)$ entries, $\beta_{\mathrm{DL},u}$ represents the large-scale attenuation of the direct BS-to-UE link, and $\mathbf{R}_{\mathrm{DL},u}\in\mathbb{C}^{M\times M}$ is its Hermitian positive semidefinite BS-side spatial correlation matrix. We adopt the normalization $\operatorname{tr}(\mathbf{R}_{\mathrm{DL},u})=M$, which preserves $\beta_{\mathrm{DL},u}$ as the average per-antenna large-scale gain. Therefore,
\begin{align}
	\mathbb{E}
	\left\{
	\mathbf{h}_{u}^{\mathrm{dl},H}
	\mathbf{h}_{u}^{\mathrm{dl}}
	\right\}
	=
	\beta_{\mathrm{DL},u}\mathbf{R}_{\mathrm{DL},u}, 
\end{align}
where $\beta_{\mathrm{DL},u}
=
\beta_{\mathrm{DL},0}
\left(\frac{d_{\mathrm{DL},u}}{d_0}\right)^{-n_{ \mathrm{DL}}}
$ 
with $d_{{\rm DL},u}$ denoting the direct BS-to-UE distance. The conventional spatially white Rayleigh model follows from $\mathbf R_{{\rm DL},u}=\mathbf I_M$, while a fully blocked direct link is represented by $\beta_{{\rm DL},u}=0$.}

\vspace{-2mm}
\subsection{End-to-End Received Signal Model}
Upon combining all the components, the signal received by UE~$u$ becomes
\begin{equation}\label{eq:yu_vector}
{\color{blue}
y_u=\mathbf h_{{\rm rad-UE},u}\mathbf W_{{\rm rad},u}\mathbf H_{{\rm sur},u}\mathbf H_{\rm BS-sur}\mathbf x+\mathbf h_u^{\rm dl}\mathbf x+n_u,}
\end{equation}
where $n_u\sim\mathcal{CN}(0,\sigma^2)$ denotes the receiver noise.\footnote{{\color{blue}This model adopts narrowband block flat fading. Let $\tau_{\mathrm{sw},u}$ and $\tau_{\mathrm{DL},u}$ denote the propagation delays of the E-FAS-assisted and direct routes, respectively, and define $\Delta\tau_u=|\tau_{\mathrm{sw},u}-\tau_{\mathrm{DL},u}|$. The coherent summation in \eqref{eq:yu_vector} is valid when the signal bandwidth $B$ satisfies $B\Delta\tau_u\ll 1$, and the delay spread within each route is negligible relative to the	symbol duration. Under this condition, both route responses are approximately constant over the occupied bandwidth, and their route phases can be absorbed into the corresponding complex channel coefficients after timing synchronization. When this condition is not satisfied, the equivalent channel becomes frequency selective. Under a single-tap representation of each route, the channel on subcarrier $k$ can be written as
\begin{align}
\mathbf h_u^{\mathrm{eq}}[k]=e^{-j2\pi f_k\tau_{\mathrm{sw},u}}\mathbf t_u^{H}+e^{-j2\pi f_k\tau_{\mathrm{DL},u}}\mathbf h_u^{\mathrm{dl},H},
\end{align}
where $f_k$ denotes the baseband frequency of subcarrier $k$. The subsequent outage, ergodic-rate, and ZF results can then be applied per subcarrier using the corresponding frequency equivalent 	channel and covariance, while the aggregate wideband rate is obtained by summing over the subcarriers. Delay spread may additionally provide frequency diversity, but requires frequency selective channel estimation and precoding.}} Substituting $\mathbf{x}=\sum_{i=1}^{K}\sqrt{P_{i}}\mathbf{w}_{i}x_{i}$ into \eqref{eq:yu_vector} gives (\ref{eq:yu_interference}) (see top of next page). 
\begin{figure*}
\begin{multline}\label{eq:yu_interference}
y_u=\sqrt{P_u}\left(
		\mathbf{h}_{\mathrm{rad\textnormal{-}UE},u}
		\mathbf{W}_{\mathrm{rad},u}
		\mathbf{H}_{\mathrm{sur},u}
		\mathbf{H}_{\mathrm{BS\textnormal{-}sur}}
		\mathbf{w}_u
		+\mathbf{h}_{u}^{\mathrm{dl}}\mathbf{w}_u
		\right)x_u\\
+\sum_{i\neq u}\sqrt{P_i}
\left(
		\mathbf{h}_{\mathrm{rad\textnormal{-}UE},u}
		\mathbf{W}_{\mathrm{rad},u}
		\mathbf{H}_{\mathrm{sur},u}
		\mathbf{H}_{\mathrm{BS\textnormal{-}sur}}
		\mathbf{w}_i
+\mathbf{h}_{u}^{\mathrm{dl}}\mathbf{w}_i\right)x_i
+ n_u
\end{multline}
\hrulefill
\end{figure*}
The summation over $i\neq u$ in \eqref{eq:yu_interference} represents the inter-user interference created by the signals intended for the remaining UEs. Upon defining the effective end-to-end channel coefficient associated with stream $x_i$ as
\begin{equation}
h_{u,i}^{\mathrm{eq}}\triangleq\mathbf{h}_{\mathrm{rad\textnormal{-}UE},u}\mathbf{W}_{\mathrm{rad},u}\mathbf{H}_{\mathrm{sur},u}\mathbf{H}_{\mathrm{BS\textnormal{-}sur}}\mathbf{w}_i+\mathbf{h}_{u}^{\mathrm{dl}}\mathbf{w}_i ,
\end{equation}
the received signal finally becomes
\begin{equation}
y_u=\sqrt{P_u}\,h_{u,u}^{\mathrm{eq}} x_u+\sum_{i\neq u}\sqrt{P_i}\,h_{u,i}^{\mathrm{eq}} x_i+ n_u,
\end{equation}
which is the complete multiuser E-FAS downlink model.

\vspace{-2mm}
{\color{blue}\section{End-to-End E-FAS Channel Statistics}\label{sec:eq_channel}
We now characterize the equivalent end-to-end channel before specifying a transmit precoder. This distinction is important because	the statistical structure of the propagation channel is a property of the E-FAS architecture and should not depend on the subsequent transmission strategy. To accommodate user specific route or surface radiation configuration, we define  	the calibrated nominal route transfer matrix as
\begin{align}
\mathbf{A}_{u}\triangleq\mathbf{W}_{\mathrm{rad},u}\overline{\mathbf{H}}_{\mathrm{sur},u}\in\mathbb{C}^{N_L\times N_s}.
\end{align}
The common route model follows by setting $\mathbf{A}_{u}=\mathbf{W}_{\mathrm{rad}}\overline{\mathbf{H}}_{\mathrm{sur}}$ for all users. Throughout the closed-form analysis, $\mathbf{A}_{u}$ denotes this calibrated nominal route transfer matrix.

For an actual perturbed surface realization, the corresponding route transfer matrix would be given by
\begin{align}
\widetilde{\mathbf{A}}_{u}\triangleq\mathbf{W}_{\mathrm{rad},u}\mathbf{H}_{\mathrm{sur},u} =\mathbf{A}_{u}+\Delta\mathbf{A}_{u},
\end{align}
where $\Delta\mathbf{A}_{u}\triangleq\mathbf{W}_{\mathrm{rad},u}\Delta\mathbf{H}_{\mathrm{sur},u}$, and $\Delta\mathbf{H}_{\mathrm{sur},u}=\Delta\mathbf{H}_{\mathrm{imp},u}+ \mathbf{H}_{\mathrm{sec},u}$ includes both the impedance- and environment-induced deviations and the secondary guided components associated with the route of UE $u$.  Proposition~\ref{pro:1} and the subsequent conditional moment derivations, to be presented, remain valid for a fixed perturbed surface realization by replacing $\mathbf{A}_{u}$ with $\widetilde{\mathbf{A}}_{u}$. If the surface perturbations vary randomly across transmission intervals, they introduce an additional slow-timescale source of uncertainty in the unconditional channel statistics.

\vspace{-2mm}
\subsection{SW-Assisted Vector Channel}
The SW-assisted BS-to-UE channel is defined as
\begin{align}
\mathbf{t}_{u}\triangleq\mathbf{h}_{\mathrm{rad\textnormal{-}UE},u}\mathbf{A}_{u}\mathbf{H}_{\mathrm{BS\textnormal{-}sur}}\in\mathbb{C}^{1\times M}.
\end{align}
Using the correlated BS-to-surface model introduced in \eqref{eq:h-corr}, this channel can be written as
\begin{align}
\mathbf{t}_{u}=\sqrt{\beta_{\mathrm{BS}}}\,\mathbf{h}_{\mathrm{rad\textnormal{-}UE},u}\mathbf{A}_{u}\mathbf{R}_{\mathrm{s}}^{1/2}\mathbf{G}_{\mathrm{BS\textnormal{-}sur}}\mathbf{R}_{\mathrm{B}}^{1/2}.
\end{align}
The complete equivalent channel vector, including the residual direct link, is then defined as
\begin{align}
\mathbf{h}_{u}^{\mathrm{eq}}\triangleq\mathbf{t}_{u}^{H}+\mathbf{h}_{u}^{\mathrm{dl},H}\in\mathbb{C}^{M\times 1}.
\end{align}

\vspace{-2mm}
\subsection{Conditional and Unconditional Channel Distribution}
In this subsection, we define $\mathbf{B}_{u}\triangleq\mathbf{A}_{u}\mathbf{R}_{\mathrm{s}}\mathbf{A}_{u}^{H}\in\mathbb{C}^{N_L\times N_L}$.  Then, the following proposition provides the conditional distribution of the SW-assisted channel. 

\begin{proposition}\label{pro:1}
Conditioned on $\mathbf{h}_{\mathrm{rad\textnormal{-}UE},u}$, the SW-assisted channel satisfies
\begin{align}
\mathbf{t}_{u}\,\big|\,\mathbf{h}_{\mathrm{rad\textnormal{-}UE},u}\sim\mathcal{CN}\left(\mathbf{0},\nu_{u}\mathbf{R}_{\mathrm{B}}\right),
\end{align}
where $\nu_{u}\triangleq\beta_{\mathrm{BS}}\,\mathbf{h}_{\mathrm{rad\textnormal{-}UE},u}\mathbf{B}_{u}\mathbf{h}_{\mathrm{rad\textnormal{-}UE},u}^{H}$.
\end{proposition}

\begin{IEEEproof}
For a fixed realization of $\mathbf{h}_{\mathrm{rad\textnormal{-}UE},u}$, define $\mathbf{b}_{u}\triangleq\mathbf{h}_{\mathrm{rad\textnormal{-}UE},u}\mathbf{A}_{u}\mathbf{R}_{\mathrm{s}}^{1/2}$. Since $\mathbf{G}_{\mathrm{BS\textnormal{-}sur}}$ contains i.i.d.~$\mathcal{CN}(0,1)$ entries, $\mathbf{b}_{u}\mathbf{G}_{\mathrm{BS\textnormal{-}sur}}$ is conditionally Gaussian with covariance $\|\mathbf{b}_{u}\|^{2}\mathbf{I}_{M}$. Right multiplication by $\mathbf{R}_{\mathrm{B}}^{1/2}$ therefore gives the conditional covariance $\beta_{\mathrm{BS}}\|\mathbf{b}_{u}\|^{2}\mathbf{R}_{\mathrm{B}}=\nu_{u}\mathbf{R}_{\mathrm{B}}$.
\end{IEEEproof}

Because $\nu_{u}$ depends on the random radiation-to-UE channel,	the unconditional distribution of $\mathbf{t}_{u}$ is generally not exactly Gaussian. It is instead a Gaussian scale mixture given by
\begin{align}
p_{\mathbf{t}_{u}}(\mathbf{x})=\int_{0}^{\infty}p_{\mathcal{CN}}\left(\mathbf{x};\mathbf{0},\nu\mathbf{R}_{\mathrm{B}}\right)f_{\nu_{u}}(\nu)\,\mathrm{d}\nu.
\end{align}
Hence, averaging the conditional covariance does not by itself imply an exactly Gaussian unconditional channel.

\vspace{-2mm}
\subsection{Equivalent Channel Covariance}
Although the unconditional distribution is generally a Gaussian scale mixture, its first two moments admit closed-form expressions. In particular, we have the results $\mathbb E\{\mathbf t_u\}=\mathbf 0$ and $\mathbb E\left\{\mathbf t_u^{H}\mathbf t_u\right\}=\Omega_{{\rm sw},u}\mathbf R_{\rm B}$, where 
\begin{equation}
\Omega_{{\rm sw},u}\triangleq\eta_{L,u}G_{{\rm rad},u}\beta_{\rm BS}\ell_{{\rm RU},u}\operatorname{tr}\left(\mathbf R_{L,u}\mathbf A_u\mathbf R_s\mathbf A_u^H\right),
\end{equation}
in which $\beta_{\rm BS}=\beta_0\left(d_{\rm BS}/d_0\right)^{-n_{\rm BS}}$ provides the single absolute free-space reference gain of the E-FAS-assisted route, while $\ell_{{\rm RU},u}=\left(d_{{\rm RU},u}/d_0\right)^{-n_{\rm RU}}$ accounts only for the relative distance dependence of the final space wave segment. As a consequence, no additional absolute free-space reference gain is introduced at the re-radiation-to-UE stage. 

Since the direct channel is independent of the SW-assisted channel, the complete equivalent channel satisfies
	\begin{align}
		\mathbf{h}_{u}^{\mathrm{eq}}
		\,\big|\,
		\mathbf{h}_{\mathrm{rad\textnormal{-}UE},u}
		\sim
		\mathcal{CN}
		\left(
		\mathbf{0},
		\nu_{u}\mathbf{R}_{\mathrm{B}}
		+
		\beta_{\mathrm{DL},u}
		\mathbf{R}_{\mathrm{DL},u}
		\right).
	\end{align}
	where its unconditional covariance is
	\begin{equation}
		\mathbf{R}_{\mathrm{eq},u}
		\triangleq
		\mathbb{E}
		\left\{
		\mathbf{h}_{u}^{\mathrm{eq}}
		\mathbf{h}_{u}^{\mathrm{eq},H}
		\right\}
		=
		\Omega_{\mathrm{sw},u}\mathbf{R}_{\mathrm{B}}
		+
		\beta_{\mathrm{DL},u}\mathbf{R}_{\mathrm{DL},u}.
	\end{equation}
This expression shows explicitly that deterministic surface routing does not, by itself, guarantee a spatially white equivalent channel. The BS-side and direct link correlation matrices remain present in the end-to-end channel covariance.

\vspace{-2mm}	
\subsection{Gaussian Approximation and Spatially White Special Case}
To quantify when the random conditional scale $\nu_{u}$ concentrates around its mean, we define
\begin{align}
\mathbf{C}_{u}\triangleq\mathbf{R}_{\mathrm{L},u}^{1/2}\mathbf{B}_{u}\mathbf{R}_{\mathrm{L},u}^{1/2}.
\end{align}
As a consequence, for the correlated Rayleigh re-radiation-to-UE channel, we have $\mathbb{E}\{\nu_u\}=\Omega_{\rm sw,u}$ and $\operatorname{Var}(\nu_u)=\eta_{L,u}^{2}G_{\rm rad,u}^{2}\beta_{\rm BS}^{2}\ell_{{\rm RU},u}^2\operatorname{tr}\!\left(\mathbf C_u^{2}\right)$. Hence, the scale variance is
\begin{align}
\frac{\operatorname{Var}(\nu_{u})}{\mathbb{E}^{2}\{\nu_{u}\}}=\frac{\operatorname{tr}(\mathbf{C}_{u}^{2})}{\operatorname{tr}^{2}(\mathbf{C}_{u})}\triangleq\frac{1}{N_{\mathrm{eff},u}},
\end{align}
where $N_{\mathrm{eff},u}\triangleq\frac{\operatorname{tr}^{2}(\mathbf{C}_{u})}{\operatorname{tr}(\mathbf{C}_{u}^{2})}$ represents the effective number of independently contributing modes.  When $N_{\mathrm{eff},u}$ is sufficiently large, $\nu_{u}$ concentrates around $\Omega_{\mathrm{sw},u}$, and the equivalent channel can be approximated as
\begin{align}
		\mathbf{h}_{u}^{\mathrm{eq}}
		\overset{\mathrm{approx.}}{\sim}
		\mathcal{CN}
		\left(
		\mathbf{0},
		\mathbf{R}_{\mathrm{eq},u}
		\right).
\end{align}
The spatially white Gaussian model is recovered under the additional conditions $
		\mathbf{R}_{\mathrm{B}}
		=
		\mathbf{I}_{M}$ and $
		\mathbf{R}_{\mathrm{DL},u}
		=
		\mathbf{I}_{M}$, 
for which
\begin{align}
		\mathbf{h}_{u}^{\mathrm{eq}}
		\overset{\mathrm{approx.}}{\sim}
		\mathcal{CN}
		\left(
		\mathbf{0},
		\Omega_{\mathrm{eq},u}\mathbf{I}_{M}
		\right),
\end{align}
with $\Omega_{\mathrm{eq},u}\triangleq\Omega_{\mathrm{sw},u}+\beta_{\mathrm{DL},u}$. Accordingly, the i.i.d.~Gaussian channel used for the subsequent closed-form performance analysis is a channel hardening and spatially white special case of the general E-FAS model, rather than a universal consequence of deterministic surface routing.

\vspace{-2mm}
\subsection{Inter-User Dependence and Route Coupling}\label{sec:inter}
Different users may share the same BS-to-surface channel realization. Consequently, mutually uncorrelated user channels do not necessarily imply statistically independent user channels. Conditioned on the radiation-to-UE channels, the cross-covariance of the SW-assisted channels of UEs $u$ and $v$ is
\begin{align}
\mathbb{E}_{\mathbf{G}_{\mathrm{BS\textnormal{-}sur}}}\left\{\mathbf{t}_{u}^{H}\mathbf{t}_{v}\,\big|\,\mathbf{h}_{\mathrm{rad\textnormal{-}UE},u},\mathbf{h}_{\mathrm{rad\textnormal{-}UE},v}\right\}=\beta_{\mathrm{BS}}\kappa_{uv}\mathbf{R}_{\mathrm{B}},\label{eq:mean}
\end{align}
where
\begin{align}
\kappa_{uv}\triangleq\mathbf{h}_{\mathrm{rad\textnormal{-}UE},v}\mathbf{A}_{v}\mathbf{R}_{\mathrm{s}}\mathbf{A}_{u}^{H}\mathbf{h}_{\mathrm{rad\textnormal{-}UE},u}^{H}.\label{eq:cov}
\end{align}
 Therefore, $\mathbf{A}_{v}\mathbf{R}_{\mathrm{s}}\mathbf{A}_{u}^{H}$ quantifies the overlap between the excitation subspaces associated with the routes of UE $u$ and UE $v$. User channel independence is justified when user specific routes occupy disjoint or approximately orthogonal excitation subspaces, the radiation-to-UE channels are independent, and common scattering between users is negligible. Otherwise, the equivalent channel matrix may exhibit both spatial correlation and inter-user dependence. The classical independent Rayleigh MU-MIMO model is therefore adopted later as a tractable special case. BS-side spatial correlation is evaluated numerically in Section~\ref{sec:numerical}, while route overlap is characterized analytically through the conditional cross-covariance in \eqref{eq:mean} and \eqref{eq:cov}.}
 

\vspace{-2mm}
{\color{blue}\section{Single-User Performance}\label{sec:single_user}
We now specialize the general E-FAS channel model to the channel hardening Gaussian approximation established in Section~\ref{sec:eq_channel}. The spatially white case is considered to obtain transparent closed-form expressions and to isolate the effect of the E-FAS average channel gain. Accordingly, $\mathbf{h}^{\mathrm{eq}}\sim\mathcal{CN}\left(\mathbf{0},\Omega_{\mathrm{eq}}\mathbf{I}_{M}\right)$ is adopted in this section. The resulting outage and ergodic capacity expressions are standard Rayleigh fading results specialized to the E-FAS parameter $\Omega_{\mathrm{eq}}$.
}

\vspace{-2mm}
{\color{blue}\subsection{CSI Assumptions and Single-User Precoding}
For single-user transmission, let $\mathbf{h}^{\mathrm{eq}}\in\mathbb{C}^{M\times 1}$ denote the equivalent end-to-end channel introduced in Section~\ref{sec:eq_channel}, and let $\mathbf{w}\in\mathbb{C}^{M\times 1}$ denote a unit-norm transmit precoder. We distinguish between two transmitter CSI regimes.

\emph{1) No instantaneous CSIT:} The BS is assumed to know the statistical channel parameters and the	nominal E-FAS route configuration, but not the instantaneous realizations of the small-scale fading segments. Accordingly, the BS uses a unit-norm precoder $\mathbf{w}_{0}$ that is independent of $\mathbf{h}^{\mathrm{eq}}$. This includes a fixed broad beam, statistical transmission, or an isotropically distributed random precoder.

Under the spatially white Gaussian specialization $\mathbf{h}^{\mathrm{eq}}\sim\mathcal{CN}\left(\mathbf{0},\Omega_{\mathrm{eq}}\mathbf{I}_{M}\right)$, the rotational invariance of the channel implies that, for any unit-norm $\mathbf{w}_{0}$ independent of $\mathbf{h}^{\mathrm{eq}}$,
\begin{align}
h\triangleq\mathbf{h}^{\mathrm{eq},H}\mathbf{w}_{0}\sim\mathcal{CN}\left(0,\Omega_{\mathrm{eq}}\right).
\end{align}

More generally, if only statistical CSIT is available under the correlated channel model, then the average received channel power for a unit-norm precoder $\mathbf{w}$ is
\begin{align}
\mathbb{E}\left\{\left|\mathbf{h}^{\mathrm{eq},H}\mathbf{w}\right|^{2}\right\}=\mathbf{w}^{H}\mathbf{R}_{\mathrm{eq}}\mathbf{w}.
\end{align}
Therefore, statistical beamforming selects the dominant eigenvector of $\mathbf{R}_{\mathrm{eq}}$ and achieves the average gain $\lambda_{\max}(\mathbf{R}_{\mathrm{eq}})$.\footnote{\color{blue}Under the spatially white specialization $\mathbf{R}_{\mathrm{eq}}=\Omega_{\mathrm{eq}}\mathbf{I}_{M}$, however, no statistically preferred transmit direction exists, and every fixed unit-norm precoder achieves the same average channel gain $\Omega_{\mathrm{eq}}$. Therefore, the isotropic-precoding assumption is not intended to represent optimal beamforming. Rather, it provides a tractable no instantaneous CSIT benchmark and is statistically equivalent, under the spatially white channel model, to any fixed unit-norm beam that is independent of the instantaneous channel.} Such a benchmark is relevant when end-to-end channel acquisition or tracking is unavailable, outdated, or incurs excessive overhead because of route reconfiguration.

\emph{2) Perfect instantaneous CSIT:} When the BS has perfect instantaneous knowledge of $\mathbf{h}^{\mathrm{eq}}$, MRT is employed according to $\mathbf{w}_{\mathrm{MRT}}=\frac{\mathbf{h}^{\mathrm{eq}}}{\left\|\mathbf{h}^{\mathrm{eq}}\right\|}$. The corresponding effective channel power becomes $\left|\mathbf{h}^{\mathrm{eq},H}\mathbf{w}_{\mathrm{MRT}}\right|^{2}=\left\|\mathbf{h}^{\mathrm{eq}}\right\|^{2}$. The outage and capacity expressions in Sections \ref{sec:out} and \ref{sec:cap} will first characterize the no instantaneous CSIT benchmark, while the MRT case will be presented separately in Section~\ref{sec:mrt}.}

\vspace{-2mm}
{\color{blue}
\subsection{Effective SNR}\label{sec-snrs}
With a single intended stream and transmit power $P$, the received signal reduces to
\begin{align}
y=\sqrt{P}hx+n,
\end{align}
where $x$ is a unit power symbol and $n\sim\mathcal{CN}(0,\sigma^2)$. Hence, the instantaneous SNR is given by
\begin{align}
\gamma=\frac{P|h|^2}{\sigma^2}=\rho |h|^2,~\mbox{where }\rho\triangleq\frac{P}{\sigma^2}.
\end{align}

To explicitly characterize the distance dependence of the E-FAS-assisted route, consider the segmented propagation geometry illustrated in Fig.~\ref{fig:efas_model_b}. In the single-user setting, let $d_{\rm BS}$, $d_{\rm sw}$, and $d_{\rm RU}$ denote the BS-to-surface, guided-SW, and surface-to-UE propagation distances, respectively. We define the end-to-end distance parameter as $D \triangleq d_{\rm BS}+d_{\rm sw}+d_{\rm RU}$. Accordingly, the guided propagation distance is $ d_{\rm sw}(D)=D-d_{\rm BS}-d_{\rm RU}$. For comparison with conventional direct transmission, $d_{\rm DL}$ denotes the corresponding Euclidean BS-to-UE distance. Its relation to $D$ depends on the deployment geometry. For the trapezium geometry in Fig.~\ref{fig:efas_model_b} and Section~\ref{sec:numerical}, we have $d_{\rm DL}(D)=D-1$. Therefore, the calibrated nominal route transfer matrix becomes distance dependent through $ \mathbf A(D) \triangleq \mathbf W_{\rm rad} \overline{\mathbf H}_{\rm sur} \big(d_{\rm sw}(D)\big)$,  and the SW-assisted channel gain can be written as 
\begin{multline}
\Omega_{\rm sw}(D) = \eta_LG_{\rm rad}\beta_0 \left(\frac{d_{\rm BS}}{d_0}\right)^{-n_{\rm BS}} \left(\frac{d_{\rm RU}}{d_0}\right)^{-n_{\rm RU}}\\
\times \operatorname{tr} \left( \mathbf R_L \mathbf A(D) \mathbf R_s \mathbf A^{H}(D) \right). 
\end{multline}
Hence, the dependence of the E-FAS-assisted link on $D$ is determined by the lengths of the two space-wave segments and, most importantly, by the guided path length $d_{\rm sw}(D)$.

Using the above distance-dependent gains, the average SNRs of direct and E-FAS-assisted transmission under the spatially white Gaussian specialization are, respectively,
\begin{align}
\bar{\gamma}_{\rm DL}=\frac{P\beta_{\rm DL}}{\sigma^2},~\mbox{and }\bar{\gamma}_{\mathrm{E\text{-}FAS}}=\frac{P(\beta_{\rm DL}+\Omega_{\rm sw})}{\sigma^2}.
\end{align}
The ratio between these average SNRs is given by
\begin{equation}\label{eq-40}
A_{\rm E\text{-}FAS}=\frac{\bar{\gamma}_{\mathrm{E\text{-}FAS}}}{\bar{\gamma}_{\rm DL}}=1+\frac{\Omega_{\rm sw}}{\beta_{\rm DL}}.
\end{equation}
For a fully blocked direct link, i.e., $\beta_{\rm DL}=0$, the ratio in \eqref{eq-40} is not defined, and the E-FAS-assisted average SNR is directly given by $\bar{\gamma}_{\rm E\text{-}FAS}=P\Omega_{\rm sw}/\sigma^2$. Therefore, within the adopted link-level model, the SW-assisted route provides a nonnegative average SNR gain. As the direct link becomes stronger, one might expect that this gain would eventually vanish. However, this is generally not true because the guided SW route typically has a much lower loss than the space wave direct link.}

\vspace{-2mm}
\subsection{Outage Probability}\label{sec:out}
Let $R_{0}$ be the target spectral efficiency in bits per second per Hertz (bps/Hz). An outage occurs whenever the instantaneous mutual information  $\log_{2}(1+\gamma)$ falls below $R_{0}$. 
{\color{blue}Under the spatially white Gaussian specialization, $|h|^{2}$ is exponentially distributed with mean $\Omega_{\mathrm{eq}}$. Applying the standard Rayleigh fading outage result \cite{tse2005fundamentals} yields the following corollary.}

\begin{corollary}[Outage Probability]
Let $h \sim \mathcal{CN}(0,\Omega_{\mathrm{eq}})$ and $\gamma=\rho |h|^{2}$. For any target rate $R_{0}>0$, the outage probability is formulated as
\begin{align}\label{eq:Pout_final}
P_{\mathrm{out}}(R_{0})= 1 -
			\exp\!\left(
			-\frac{\gamma_{0}}{\rho \Omega_{\mathrm{eq}}}
			\right),
\end{align}
where $\gamma_{0} = 2^{R_{0}} - 1$.
\end{corollary}

\begin{IEEEproof}
By definition, we have
\begin{equation}
P_{\mathrm{out}}(R_{0})= \Pr\{\gamma < \gamma_{0}\}= \Pr\left\{|h|^{2} < \frac{\gamma_{0}}{\rho}\right\}.
\end{equation}
As $|h|^{2}$ is exponential with mean $\Omega_{\mathrm{eq}}$, its cumulative distribution function (CDF) is expressed as
$
F_{|h|^{2}}(x)=1 -e^{-\frac{x}{\Omega_{\mathrm{eq}}}}$ for $x\ge 0$, which, when evaluated at $x=\gamma_{0}/\rho$, yields \eqref{eq:Pout_final}.
\end{IEEEproof}

It is worth pointing out that the influence of the E-FAS environment is characterized entirely by the parameter $\Omega_{\mathrm{eq}}$, which encapsulates the deterministic SW propagation effects and any remaining direct path component.

\vspace{-2mm}
\subsection{Ergodic Capacity}\label{sec:cap}
The ergodic capacity is defined as
\begin{align}
C=\mathbb{E}\!\left\{\log_{2}(1+\gamma)\right\}=\mathbb{E}\!\left\{\log_{2}(1+\rho |h|^{2})\right\}.
\end{align}
Given that $|h|^{2}$ is exponentially distributed, the integral admits a well-known closed form. {\color{blue}As a result, applying the standard exponential fading ergodic capacity integral to the E-FAS variance $\Omega_{\mathrm{eq}}$ yields the following specialization.}

\begin{corollary}[Ergodic Capacity]
For $h \sim \mathcal{CN}(0,\Omega_{\mathrm{eq}})$ and $\gamma=\rho |h|^{2}$, the ergodic capacity of the  single-user E-FAS considered is obtained as
\begin{align}\label{eq:Cerg_final}
C=\frac{1}{\ln 2}\exp\!\left(\frac{1}{\rho \Omega_{\mathrm{eq}}}\right)\mathrm{E}_{1}\!\left(\frac{1}{\rho \Omega_{\mathrm{eq}}}\right),
\end{align}
where $\mathrm{E}_{1}(\cdot)$ denotes the exponential integral.
\end{corollary}

\begin{IEEEproof}
Since the SNR $\gamma$ follows the exponential distribution, the corresponding probability density function (PDF) is written as $f_{\gamma}(x)=\frac{1}{\rho \Omega_{\mathrm{eq}}}\exp\!\left(-\frac{x}{\rho\Omega_{\mathrm{eq}}}\right)$. Thus, we have
\begin{align}
C=\frac{1}{\rho\Omega_{\mathrm{eq}}\ln 2}\int_{0}^{\infty}\ln(1+x)\exp\!\left(-\frac{x}{\rho\Omega_{\mathrm{eq}}}\right) dx.
\end{align}
Using the identity of $\int_{0}^{\infty}\ln(1+x) e^{-ax} dx=\frac{1}{a} e^{a} E_{1}(a)$, for $a>0$, associated with $a = 1/(\rho \Omega_{\mathrm{eq}})$ yields \eqref{eq:Cerg_final}.
\end{IEEEproof}

\vspace{-2mm}
{\color{blue}\subsection{MRT Benchmark with Instantaneous CSIT}\label{sec:mrt}
We now consider perfect instantaneous end-to-end CSIT and MRT. Under 	the spatially white Gaussian specialization,
\begin{align}
		\left\|
		\mathbf{h}^{\mathrm{eq}}
		\right\|^{2}
		\sim
		\operatorname{Gamma}
		\left(
		M,\Omega_{\mathrm{eq}}
		\right),
	\end{align}
where $M$ is the shape parameter and $\Omega_{\mathrm{eq}}$ is the scale parameter. Therefore, the instantaneous MRT SNR is 
$
		\gamma_{\mathrm{MRT}}
		=
		\rho
		\left\|
		\mathbf{h}^{\mathrm{eq}}
		\right\|^{2}$. 
Now, for a target rate $R_{0}$ and $\gamma_{0}=2^{R_{0}}-1$, the corresponding outage probability is given by
\begin{align}
		P_{\mathrm{out}}^{\mathrm{MRT}}(R_{0})
		=
		1
		-
		\exp
		\left(
		-\frac{\gamma_{0}}
		{\rho\Omega_{\mathrm{eq}}}
		\right)
		\sum_{k=0}^{M-1}
		\frac{1}{k!}
		\left(
		\frac{\gamma_{0}}
		{\rho\Omega_{\mathrm{eq}}}
		\right)^{k},
	\end{align}
	and the ergodic capacity under MRT is derived as
	\begin{align}
		C_{\mathrm{MRT}}
		=
		\frac{
			\exp
			\left(
			\frac{1}
			{\rho\Omega_{\mathrm{eq}}}
			\right)
		}{
			\ln 2
		}
		\sum_{n=1}^{M}
		E_{n}
		\left(
		\frac{1}
		{\rho\Omega_{\mathrm{eq}}}
		\right),
	\end{align}
	where
$
		E_{n}(x)
		\triangleq
		\int_{1}^{\infty}
		\frac{e^{-xt}}{t^{n}}
		\,\mathrm{d}t
$ is the generalized exponential integral. Compared with the no instantaneous CSIT benchmark, MRT increases the average effective channel power from $\Omega_{\mathrm{eq}}$ to $M\Omega_{\mathrm{eq}}$, providing the conventional transmit array gain.}

\vspace{-2mm}
\subsection{High-SNR Asymptotics}
The behavior at high SNR provides insight into both the diversity order and coding gain. As $\rho \to \infty$, $1/(\rho\Omega_{\mathrm{eq}})\to 0$, and the exponential integral admits the approximation
\begin{equation}
E_{1}(x)=-\gamma_{E}- \ln(x)+ x+ o(x),
\end{equation}
where $\gamma_{E}$ is Euler's constant. Substituting this into \eqref{eq:Cerg_final} gives
\begin{align}
C\approx\frac{1}{\ln 2}\left[-\gamma_{E}- \ln\!\left(\frac{1}{\rho\Omega_{\mathrm{eq}}}\right)\right]
=\log_{2}(\rho \Omega_{\mathrm{eq}})-\frac{\gamma_{E}}{\ln 2}.\label{eq:c-approx}
\end{align}
The parameter $\Omega_{\mathrm{eq}}$ appears multiplicative inside the logarithm, indicating that improvements in SW-guided gain through enhanced routing, surface radiation configuration, or reduced attenuation, translate directly into additional coding gain.

{\color{blue}The unit pre-log in \eqref{eq:c-approx} represents the single stream high-SNR capacity slope. Separately, expanding \eqref{eq:Pout_final} for large $\rho$ gives $P_{\mathrm{out}}(R_0)\sim\frac{\gamma_0}{\rho\Omega_{\mathrm{eq}}}$, $ \rho\rightarrow\infty$, which confirms a diversity order of one. Hence, E-FAS improves the outage coding gain and the high-SNR capacity offset through $\Omega_{\mathrm{eq}}$, while leaving the diversity order and the single stream pre-log unchanged.}


{\color{blue}For convenience, define $\Xi \triangleq G_{\rm rad}\beta_{\rm BS}\ell_{{\rm RU}}\operatorname{tr}\!\left(\mathbf R_{\rm L}\mathbf A\mathbf R_{\rm s}\mathbf A^{H}\right)$, such that $\Omega_{\rm sw}=\eta_L\Xi$. To explicitly characterize the impact of the surface-to-space conversion efficiency, define the power normalized outage coding gain factor as
\begin{align}
G_c(\eta_L)\triangleq\frac{\Omega_{\rm eq}}{\sigma^2}=\frac{\beta_{\rm DL}+\eta_L\Xi}{\sigma^2},
\end{align}
Thus, the high-SNR outage probability can be found as
\begin{align}
P_{\mathrm{out}}(R_0)\sim\frac{\gamma_0}{P\mathcal G_{\mathrm c}(\eta_{\mathrm L})}.
\end{align}
As such, the surface-to-space conversion efficiency affects the outage coding gain but does not alter the diversity order. In other words, increasing $\eta_L$ strengthens the SW component and improves the coding gain but the diversity order is unchanged.}

{\color{blue}For MRT, the small argument expansion of the Gamma CDF for $\rho\rightarrow\infty$ yields
\begin{align}
P_{\mathrm{out}}^{\mathrm{MRT}}(R_{0})\sim\frac{1}{M!}\left(\frac{\gamma_{0}}{\rho\Omega_{\mathrm{eq}}}\right)^{M}.
\end{align}
Hence, MRT achieves diversity order of $M$, in contrast to the no instantaneous CSIT benchmark. Its high-SNR capacity also satisfies
\begin{align}
C_{\mathrm{MRT}}\approx\log_{2}\left(\rho\Omega_{\mathrm{eq}}\right)+\frac{\psi(M)}{\ln 2},
\end{align}
in which $\psi(\cdot)$ denotes the digamma function. As a result, both transmission strategies benefit from the E-FAS-induced gain $\Omega_{\mathrm{eq}}$, while MRT additionally exploits the BS array through transmit-array and diversity gains.}

\vspace{-2mm}
\section{Multiuser E-FAS using ZF Precoding}\label{sec:multi-user}
Here, we use ZF precoding to obtain a tractable closed-form benchmark for the E-FAS-assisted multiuser downlink. 

\vspace{-2mm}
{\color{blue}\subsection{Equivalent Channel Matrix}
For UE $u$, let $\mathbf{h}_{u}^{\mathrm{eq}}\in\mathbb{C}^{M\times 1}$ denote the equivalent channel. In the general E-FAS model, its covariance is $\mathbb{E}\left\{\mathbf{h}_{u}^{\mathrm{eq}}\mathbf{h}_{u}^{\mathrm{eq},H}\right\}=\mathbf{R}_{\mathrm{eq},u}$. Stacking the user channels gives
\begin{align}
\mathbf{H}_{\mathrm{eq}}=\left[\mathbf{h}_{1}^{\mathrm{eq}},\ldots,\mathbf{h}_{K}^{\mathrm{eq}}\right]\in\mathbb{C}^{M\times K}.
\end{align}
In general, spatial correlation, random scale variations, and route overlap prevent the entries or columns of $\mathbf{H}_{\mathrm{eq}}$ from being i.i.d.~and therefore, the post-ZF SINR does not generally follow the classical Gamma distribution. To obtain a tractable analytical benchmark, we subsequently consider the spatially white Gaussian special case. Specifically, we assume that the conditional scale factors have hardened, the effective user routes are mutually independent, and $\mathbf{R}_{\mathrm{eq},u}=\Omega_{\mathrm{eq},u}\mathbf{I}_{M}$. For the homogeneous case, $\Omega_{\mathrm{eq},u}=\Omega_{\mathrm{eq}}$ for all $u$, and hence,
\begin{align}
\operatorname{vec}\left(\mathbf{H}_{\mathrm{eq}}\right)\overset{\mathrm{approx.}}{\sim}\mathcal{CN}\left(\mathbf{0},\Omega_{\mathrm{eq}}\mathbf{I}_{MK}\right).
\end{align}
	
The following analysis for ZF precoding is restricted to this special case while BS-side spatially correlated channels are also evaluated but numerically in	Section~\ref{sec:numerical}.}

\vspace{-2mm}
\subsection{ZF Precoding and Per-User SNR}
Let the total BS transmit power denote $P$, equally allocated among the $K$ UEs so that $P_u = P/K$. With ZF precoding, the BS employs a precoder matrix
\begin{align}
\mathbf{W}=[\mathbf{w}_1,\dots,\mathbf{w}_K]=\mathbf{H}_{\mathrm{eq}}\big(\mathbf{H}_{\mathrm{eq}}^{H}\mathbf{H}_{\mathrm{eq}}\big)^{-1}\mathbf{D},
\end{align}
where the diagonal matrix $\mathbf{D}$ is chosen to satisfy a desired power normalisation. Here, we focus on column-normalized ZF precoding, where each precoding vector satisfies $\|\mathbf{w}_u\|^{2}=1$. In the ideal ZF regime, the inter-user interference is perfectly cancelled, and the  signal received at UE~$u$ reduces to
\begin{align}
y_u=\sqrt{P_u}\,\mathbf{h}_u^{\mathrm{eq},H}\mathbf{w}_u x_u+  n_u ,
\end{align}
so that the instantaneous post-ZF SNR is given by
\begin{align}\label{eq:SNR-zf}
\gamma_u^{\mathrm{ZF}}=\frac{P_u|\mathbf{h}_u^{\mathrm{eq},H}\mathbf{w}_u|^{2}}{\sigma^{2}}.
\end{align}

It is useful to express (\ref{eq:SNR-zf}) in terms of the Gram matrix $\mathbf{G}=\mathbf{H}_{\mathrm{eq}}^{H}\mathbf{H}_{\mathrm{eq}}\in \mathbb{C}^{K\times K}$. Under conventional ZF with column normalization, and for i.i.d.~Rayleigh fading with variance $\Omega_{\mathrm{eq}}$, the post-ZF SNR can be written in distribution as
\begin{align}\label{eq:ZF_SNR_gram}
\gamma_u^{\mathrm{ZF}}=\frac{P/K}{\sigma^2[\mathbf G^{-1}]_{u,u}}\stackrel{d}{=}\frac{P\Omega_{\mathrm{eq}}}{K\sigma^2}\frac{1}{[(\mathbf G/\Omega_{\mathrm{eq}})^{-1}]_{u,u}}.
\end{align}

\begin{corollary}[ZF SNR Distribution]
{\color{blue}Under the spatially white Gaussian channel specialization, column normalized ZF precoding, $M\geq K$, and equal power allocation, the per-user post-ZF SNR is Gamma distributed \cite{Spencer-2004}, i.e.,}
\begin{align}
f_{\gamma_u^{\mathrm{ZF}}}(x)=\frac{1}{\Gamma(m)\theta^{m}}x^{m-1} e^{-x/\theta},~x \ge 0,
\end{align}
where the shape parameter is $m = M-K+1$ and the scale parameter is given by $\theta=\frac{P\Omega_{\mathrm{eq}}}{K\sigma^{2}}=\frac{\rho\Omega_{\mathrm{eq}}}{K}$. 
\end{corollary}


\begin{IEEEproof}
For an $M\times K$ matrix with i.i.d.~$\mathcal{CN}(0,\Omega_{\mathrm{eq}})$ entries and column-normalized ZF precoding, the post-ZF SNR of each UE is proportional to the inverse of a diagonal element of $(\mathbf{H}_{\mathrm{eq}}^{H}\mathbf{H}_{\mathrm{eq}})^{-1}$.  It is well known that the reciprocal of this quantity follows a chi-squared distribution with $2(M-K+1)$ degree of freedom (DoF), which corresponds to a Gamma distribution with shape parameter $M-K+1$. The scale parameter is determined by the per-user transmit power $P/K$, the receiver noise variance $\sigma^2$, and the channel variance $\Omega_{\mathrm{eq}}$. As a consequence, the stated result follows.
\end{IEEEproof}

\vspace{-2mm}
\subsection{ZF Ergodic Sum-Rate}
The ergodic sum-rate under ZF precoding is given by
\begin{equation}
R_{\mathrm{sum}}^{\mathrm{ZF}}=\sum_{u=1}^{K}\mathbb{E}\left\{\log_{2}\left(1+\gamma_u^{\mathrm{ZF}}\right)\right\}.
\end{equation}
In the homogeneous case, all UEs share the same distribution, and the expression simplifies to
\begin{align}
R_{\mathrm{sum}}^{\mathrm{ZF}}=K\,\mathbb{E}\left\{\log_{2}\left(1+\gamma^{\mathrm{ZF}}\right)\right\},\label{eq:r}
\end{align}
where $\gamma^{\mathrm{ZF}}$ denotes a generic  random gamma variable having the PDF given above.


{\color{blue}For the integer shape parameter $m=M-K+1$, the expectation in \eqref{eq:r} admits the standard closed-form representation
\begin{align}\label{eq-zf}
R_{\mathrm{sum}}^{\mathrm{ZF}}=\frac{K}{\ln 2}\exp\left(\frac{1}{\theta}\right)\sum_{n=1}^{m}E_n\left(\frac{1}{\theta}\right).
\end{align}
For more intuition, Jensen's inequality gives the upper bound
\begin{align}
R_{\mathrm{sum}}^{\mathrm{ZF}}\leq K\log_2\left(1+(M-K+1)\frac{P\Omega_{\mathrm{eq}}}{K\sigma^2}\right).\label{eq:Rsum_ZF_approx}
\end{align}
}

The upper bound in \eqref{eq:Rsum_ZF_approx} becomes tighter as the post-ZF SINR concentrates around its mean, particularly when the antenna surplus $M-K$ is large. It provides a simple interpretation of the dependence of the sum-rate on $\Omega_{\mathrm{eq}}$ and the available spatial DoF. Also, \eqref{eq:Rsum_ZF_approx} reveals that the impact of E-FAS on the sum-rate is captured entirely through the effective channel variance $\Omega_{\mathrm{eq}}$, which reflects the SW-assisted gain and the distance-dependent guided-SW propagation behavior. For a fixed $\Omega_{\mathrm{eq}}$ and sufficiently large antenna surplus $M-K$, the sum-rate scales approximately linearly with $K$ at high SNR, consistent with classical ZF-based MU-MIMO systems.

\vspace{-2mm}
\section{Numerical Results}\label{sec:numerical}

{\color{blue}In this section, we provide Monte Carlo simulation results to evaluate the performance of the proposed E-FAS-assisted wireless system and compare the results with the analytical expressions derived. Unless otherwise stated, we set $M=16$, $K=4$, $N_s=N_L=64$, and $R_0=1$~bps/Hz. To provide a physically interpretable discussion, the carrier frequency is set to $f_c=26.4$~GHz, which corresponds to the operating point of the experimental E-FAS prototype in \cite{LiuEFAS2026}. Hence, $\lambda=c/f_c$ and the  free-space reference gain is $\beta_0=\left(\frac{\lambda}{4\pi }\right)^2$, after setting $d_0=1~{\rm m}$. The measured maximum realized radiation gain of the E-FAS prototype at this frequency is $G_{\rm real}=18.5$~dBi \cite{LiuEFAS2026}. Furthermore, the guided SW pathway is characterized by $n_{\rm sw}=0$ and $\kappa_{\rm sw}=0.1$~dB/m \cite{LiuEFAS2026}. 

For the E-FAS geometry, we set $d_{\rm BS}=d_{\rm RU}=1$ m. Using the general relation $d_{\rm sw}(D)=D-d_{\rm BS}-d_{\rm RU}$ in Section \ref{sec-snrs}, the guided-SW distance specializes to $d_{\rm sw}(D)=D-2$~m. Importantly, the air-SW-air route is treated as a continuous	segmented propagation path. Accordingly, only one absolute free-space reference gain, $\beta_0$, is introduced for the overall E-FAS-assisted route, while the short air segments contribute only their relative distance dependence. Specifically, $L_{\rm sw}(D)=10^{-\kappa_{\rm sw}d_{\rm sw}(D)/10}$, and under the experimentally calibrated numerical specialization, the general distance-dependent SW-assisted gain introduced in Section \ref{sec-snrs} reduces to
\begin{equation}
\Omega_{\rm sw}(D)=\beta_0\left(\frac{d_{\rm BS}}{d_0}\right)^{-n_{\rm BS}}\eta_{\rm c}G_{\rm real}L_{\rm sw}(D)\left(\frac{d_{\rm RU}}{d_0}\right)^{-n_{\rm RU}},
\end{equation}
where $\eta_{\rm c}$ denotes the space-wave-to-SW conversion efficiency. In the numerical results, we consider the ideal aligned conversion case $\eta_{\rm c}=1$ and set $n_{\rm BS}=n_{\rm RU}=2$. Because $d_{\rm BS}=d_{\rm RU}=d_0=1$~m, the relative air segment factors equal unity, yielding $\Omega_{\rm sw}(D)=\beta_0 G_{\rm real}L_{\rm sw}(D)$. The nominal route matrix is power-normalized such that $\operatorname{tr}(\mathbf R_L\bar{\mathbf A}\mathbf R_s\bar{\mathbf A}^{H})=1$, with $\mathbf A(D)=\sqrt{L_{\rm sw}(D)}\bar{\mathbf A}$, so that no additional array power scaling is introduced. Also, the receiver bandwidth is set to $B=100$~MHz with a $7$~dB receiver noise figure, giving a thermal noise power of approximately $-87$~dBm. Hence, the numerical results are parameterized directly by the BS transmit power $P$ in dBm, while the analytical SNR variable remains $\rho=P/\sigma^2$. The residual direct BS-to-UE link is evaluated using the 3GPP indoor-hotspot (InH) model~\cite{3gpp}. 

Under the adopted geometry, the direct BS-to-UE distance is $d_{\rm DL}(D)=D-1$. With $d_{\rm DL}$ in meters and $f_c$ in GHz, the LoS and non-LoS (NLoS) path losses are, respectively, 
\begin{align}
PL_{\rm LOS}(d_{\rm DL})=&~32.4+17.3\log_{10}(d_{\rm DL})+20\log_{10}(f_c),\notag\\
PL_{\rm NLOS}(d_{\rm DL})=&~\max\{PL_{\rm LOS}(d_{\rm DL}),17.30\notag\\
&~+38.3\log_{10}(d_{\rm DL})+24.9\log_{10}(f_c)\}.\notag
\end{align}
Hence, the residual NLoS direct-link gain used in the numerical evaluation is $\beta_{\rm DL}(D)=10^{-PL_{\rm NLOS}(d_{\rm DL}(D))/10}$. The LoS model is additionally used as an unobstructed direct LoS propagation benchmark where indicated.}

{\color{blue}
Fig.~\ref{fig:val-efas} validates the Gaussian scale-mixture characterization of the end-to-end E-FAS channel and examines the accuracy of its Gaussian approximation under channel hardening. For distributional validation, we set $N_s=N_L$, $\mathbf{A}=\mathbf{I}_{N_L}/\sqrt{N_L}$, and $\mathbf{R}_{\mathrm{s}}=\mathbf{R}_{\mathrm{L}}=\mathbf{R}_{\mathrm{B}}=\mathbf{R}_{\mathrm{DL}}=\mathbf{I}$. Under this unit-power normalization, with $\eta_LG_{\rm rad}=1$, the effective number of independently contributing modes reduces to $N_{\rm eff}=N_L$. Since Fig.~\ref{fig:val-efas} considers the normalized end-to-end channel power, the common large-scale propagation factor does not affect the resulting distribution and is therefore omitted from this validation experiment. The Gaussian scale-mixture representation closely matches the empirical CDF for all considered values of $N_{\rm eff}$, confirming the conditional Gaussian structure with a random variance. By contrast, the Gaussian approximation exhibits a visible mismatch for $N_{\rm eff}=4$, but the discrepancy decreases markedly as $N_{\rm eff}$ increases and becomes negligible for $N_{\rm eff}=64$. This behavior reflects the concentration of the random channel power scale around its mean as the number of effective modes grows, thereby supporting the spatially white Gaussian approximation in the channel hardening regime.
}

\begin{figure*}[]
\centering
\includegraphics[width=\linewidth]{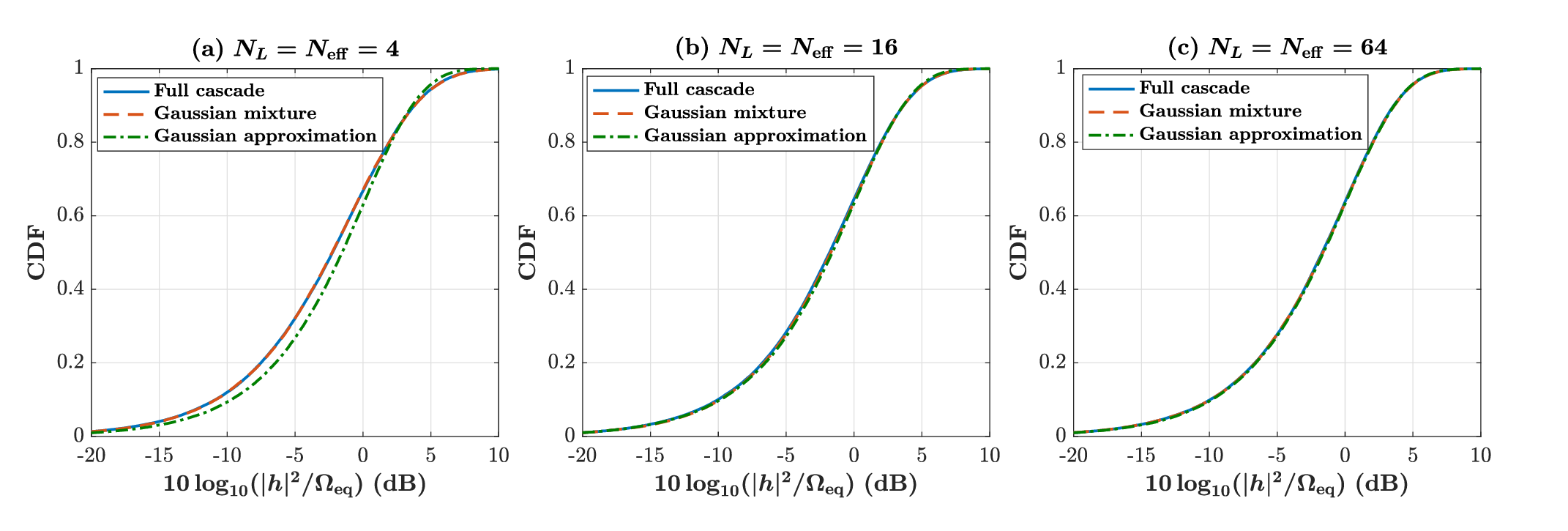}
\vspace{-6mm}
\caption{{\color{blue}CDF of the normalized end-to-end E-FAS channel power obtained from the full cascaded model, the Gaussian scale-mixture representation, and the Gaussian approximation for: (a) $N_L=N_{\mathrm{eff}}=4$, (b) $N_L=N_{\mathrm{eff}}=16$, and (c) $N_L=N_{\mathrm{eff}}=64$.}}\label{fig:val-efas}
\vspace{-2mm}
\end{figure*}

{\color{blue}Fig.~\ref{fig:corr} investigates the effect of BS-side spatial correlation using samples generated from the full cascade E-FAS channel. The direct LoS path is blocked in this experiment so that the observed behavior is determined exclusively by the E-FAS-assisted route. In this case, the BS-side covariance is modeled as $[\mathbf R_B]_{m,n}=r^{|m-n|}$ for $r\in\{0,0.4,0.7\}$, while the guided route follows the segmented E-FAS propagation model and the experimentally reported $0.1$~dB/m SW attenuation \cite{LiuEFAS2026}.

As shown in Fig.~\ref{fig:corr}(a), the uncorrelated full cascade channel closely follows the Gamma distribution predicted by the spatially white Gaussian equivalent model, with the small residual discrepancy arising from the finite mode compound Gaussian nature of the complete cascade. Increasing $r$ shifts the post-ZF SINR distribution toward lower values, demonstrating that BS-side correlation directly reduces the effective interference suppression capability of ZF. The same behavior is reflected in Fig.~\ref{fig:corr}(b), where increasing spatial correlation reduces the average ZF sum-rate over the entire transmit power range. The underlying mechanism is clarified by the results in Fig.~\ref{fig:corr}(c). The effective spatial rank $M_{\rm eff}=(\operatorname{tr}\mathbf R_B)^2/\operatorname{tr}(\mathbf R_B^2)$ decreases monotonically as $r$ increases, indicating a progressive reduction in the number of usable spatial DoF. Consequently, the post-ZF array gain and interference suppression capability deteriorate even though the large-scale E-FAS propagation gain is unchanged. Therefore, these results  separate the propagation advantage provided by E-FAS from the spatial rank penalty introduced by correlated BS-side fading.}

\begin{figure*}[]
\centering
\includegraphics[width=\linewidth]{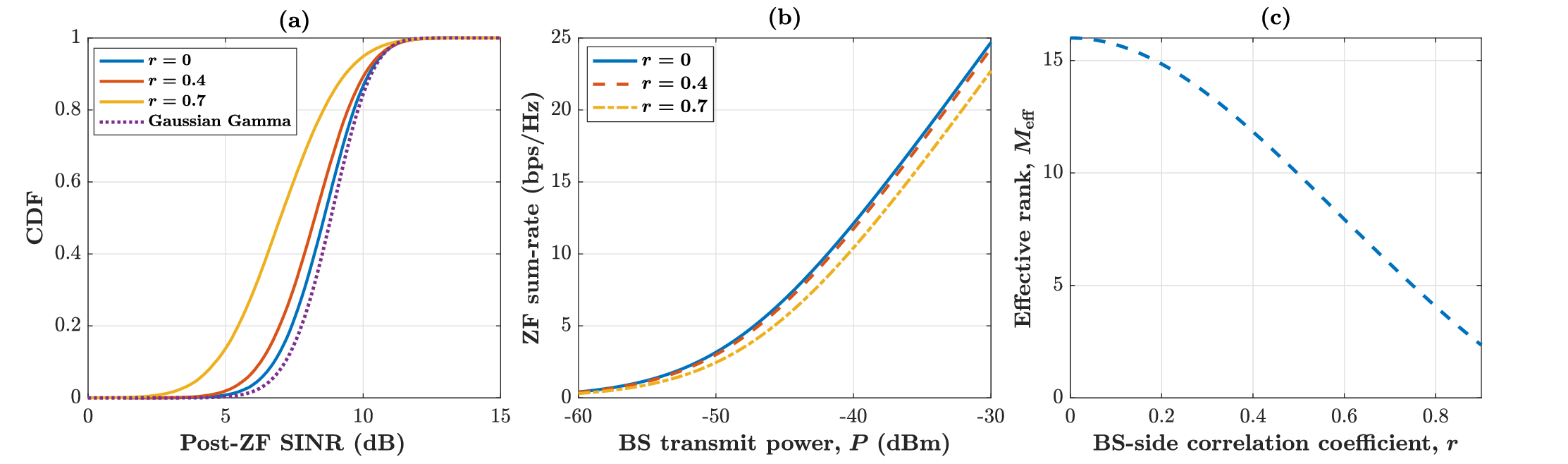}
\vspace{-5mm}
\caption{{\color{blue}Impact of BS-side spatial correlation on the full cascade E-FAS ZF performance at $D=10$~m under the segmented $26.4$~GHz channel model. (a) Empirical CDF of the post-ZF SINR at $P=-40$~dBm for different BS-side correlation coefficients, together with the Gamma benchmark predicted by the spatially white Gaussian equivalent channel. (b) Average ZF sum-rate versus BS transmit power. (c) Effective spatial rank of the BS-side correlation matrix versus the correlation coefficient $r$.}}\label{fig:corr}
\end{figure*}

{\color{blue}Fig.~\ref{fig:outage} evaluates the single-user outage performance 	of direct NLoS transmission and E-FAS-assisted NLoS transmission under the adopted trapezium segmented-path geometry, where $d_{\rm DL}=D-1$ and $d_{\rm sw}=D-2$. The direct link follows the 3GPP InH NLoS model, whereas the E-FAS-assisted link contains the same NLoS component together with the additional guided-SW route. Hence, both schemes are evaluated under the same underlying propagation condition, and the performance difference directly reflects the additional propagation gain provided by E-FAS. The analytical outage probability from~\eqref{eq:Pout_final} closely follows the Monte Carlo results over the resolvable probability range, validating the corresponding equivalent channel characterization. As shown in Fig.~\ref{fig:outage}(a), the E-FAS-assisted link enters the reliable transmission regime at a substantially lower BS transmit power than direct NLoS transmission. This large horizontal separation originates from replacing most of the end-to-end space wave propagation by the low-loss guided-SW segment. Fig.~\ref{fig:outage}(b) further demonstrates the different distance sensitivities of the two links. The outage probability of direct NLoS transmission increases rapidly with distance and approaches unity over the medium and long distance regime, whereas E-FAS maintains an outage probability on the order of $10^{-5}$ over the entire investigated range. Therefore, these results illustrate that the propagation advantage of E-FAS translates directly into a pronounced reliability improvement, particularly when conventional NLoS propagation becomes strongly distance-limited.}

\begin{figure*}[]
\centering
\includegraphics[width=.8\linewidth]{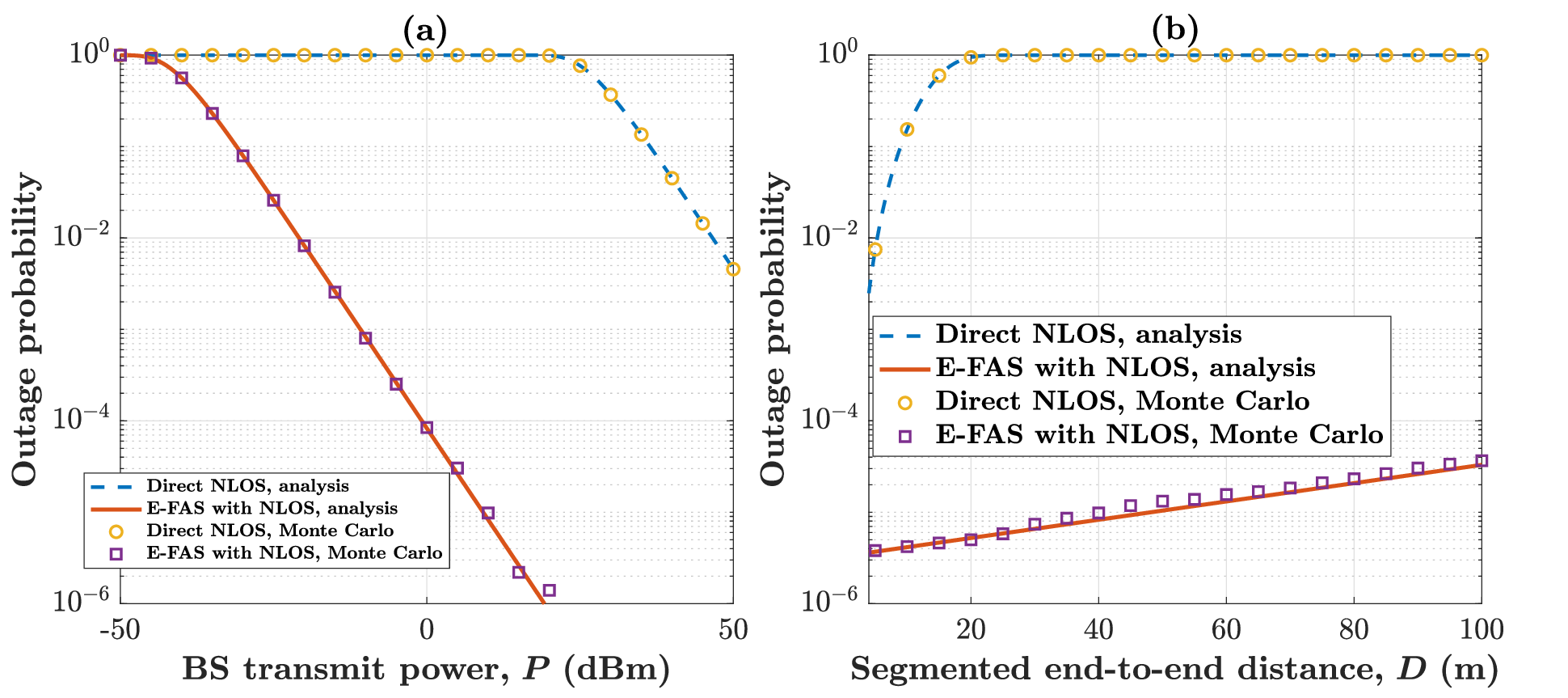}
\caption{{\color{blue}Single-user outage performance under the  $26.4$~GHz segmented E-FAS channel model. (a) Outage probability versus BS transmit power at $D=40$~m. (b) Outage probability versus the end-to-end distance at $P=10$~dBm.}}\label{fig:outage}
\end{figure*}

In Fig.~\ref{fig:capacity}, we evaluate the single-user ergodic performance of direct NLoS and E-FAS-assisted NLoS transmission under the trapezium segmented-path geometry with $d_{\rm DL}=D-1$. The analytical results obtained from~\eqref{eq:Cerg_final} are in excellent agreement with the Monte Carlo results, confirming the accuracy of the equivalent Gaussian channel characterization. As shown in Fig.~\ref{fig:capacity}(a), E-FAS provides a pronounced capacity advantage over conventional NLoS transmission over the entire BS transmit power range. At $D=40$~m, the conventional NLoS link remains strongly power-limited, whereas E-FAS achieves a high spectral efficiency even at relatively modest transmit powers because most of the end-to-end propagation distance is transferred to the low-loss guided-SW pathway. For example, at $P=40$~dBm, the E-FAS-assisted link achieves approximately $26$~bps/Hz, compared with about $3.7$~bps/Hz for the direct NLoS transmission. Fig.~\ref{fig:capacity}(b) further reveals that the E-FAS rate advantage initially increases with distance because the conventional NLoS component attenuates substantially faster than the guided-SW route. The gain reaches approximately $19$, $22$, and $24$~bps/Hz for $P=20$, $30$, and $40$~dBm, respectively, before exhibiting a gradual decrease at longer distances. This reduction occurs after the direct NLoS contribution has already become negligible, such that the remaining distance dependence is governed primarily by the finite attenuation of the guided-SW path itself. Consequently, E-FAS maintains a substantial spectral efficiency advantage throughout the entire investigated distance range.

\begin{figure*}[]
\centering
\includegraphics[width=.8\linewidth]{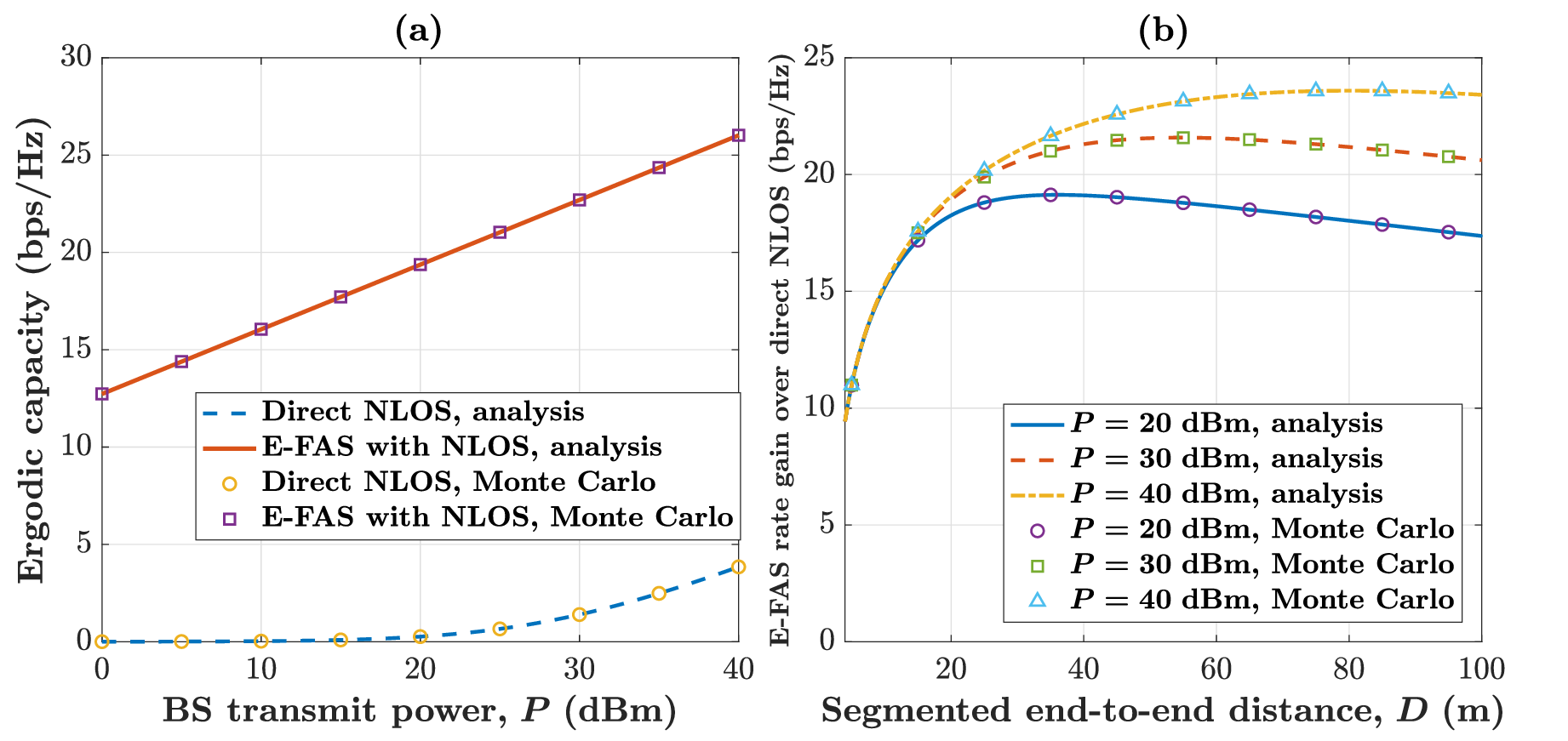}
\caption{{\color{blue}Single-user ergodic performance under the  segmented E-FAS channel model at $26.4$~GHz. (a) Ergodic capacity versus BS transmit power at $D=40$~m for direct NLoS and E-FAS-assisted NLoS transmission. (b) E-FAS ergodic-rate gain over direct NLoS transmission versus the end-to-end distance for different BS transmit powers.}}\label{fig:capacity}
\end{figure*}

{\color{blue}In Fig.~\ref{fig:ris}, we compare E-FAS with conventional direct LoS transmission and ideal passive RIS-assisted transmission. The RIS is symmetrically placed with respect to the BS and UE and follows the far-field beamforming model in~\cite{tang2022path}, with half-wavelength element spacing, unit reflection amplitude, perfect phase alignment, and broadside illumination. Under the adopted trapezium geometry, see Fig.~\ref{fig:efas_model_b}, $d_{\rm DL}=D-1$, while the two RIS propagation distances are $d_{\rm BR}=d_{\rm RU,RIS}=\sqrt{\left(\frac{\sqrt{3}}{2}\right)^2+\left(\frac{D-1}{2}\right)^2}$. In contrast to the RIS path, which contains two space wave propagation segments, the E-FAS route transfers the dominant propagation distance through the guided SW and introduces only a single free-space reference gain for the continuous air-SW-air propagation path. The measured E-FAS realized radiation gain of $18.5$~dBi and guided attenuation of $0.1$~dB/m are used directly as before.

\begin{figure*}[]
\centering
\includegraphics[width=\linewidth]{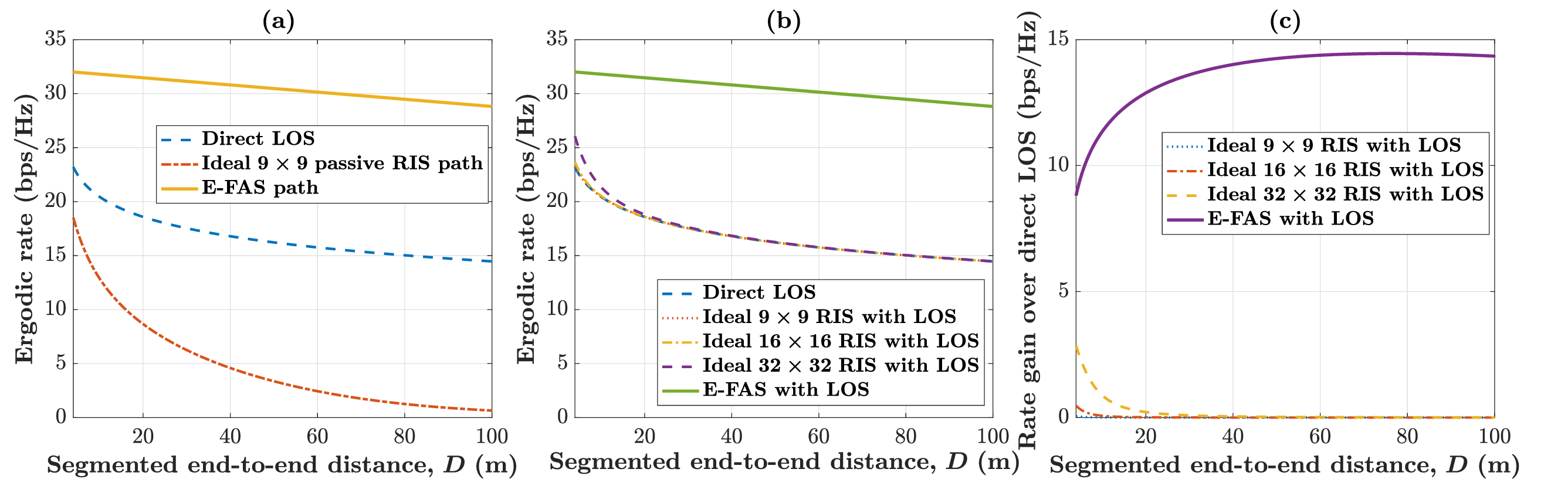}
\vspace{-4mm}
\caption{{\color{blue}Single-user MRT performance at $26.4$~GHz and $P=40$~dBm. (a) Ergodic rate comparison of the conventional direct LoS path, an ideal $9\times9$ passive-RIS path, and the E-FAS propagation path considered individually.  (b) System-level comparison under a common direct LoS component for ideal passive RISs of different sizes and E-FAS.  (c) Corresponding ergodic rate gain over conventional direct LoS.}}\label{fig:ris}
\end{figure*}

As illustrated in Fig.~\ref{fig:ris}(a), the E-FAS propagation path itself greatly outperforms not only the ideal $9\times9$ passive RIS path but also the conventional unobstructed direct LoS path over the entire investigated distance range. For example, at $D=50$~m, the E-FAS path achieves approximately $30.5$~bps/Hz, compared with $16.2$~bps/Hz for direct LoS and only $3.4$~bps/Hz for the isolated $9\times9$ RIS path. This pronounced advantage results from replacing most of the conventional space wave propagation distance by the low-loss guided-SW segment. By comparison, the passive RIS contribution contains the product of the BS-RIS and RIS-UE propagation losses and therefore decays much more rapidly with distance.

Fig.~\ref{fig:ris}(b) further compares the architectures when the same direct LoS component is available to all schemes. The results show that even when the ideal RIS aperture is increased from $9\times9$ to $32\times32$ elements, its additional contribution rapidly becomes negligible relative to the direct LoS component, whereas E-FAS maintains a substantial separation from all passive RIS benchmarks. At $D=50$~m, the $32\times32$ RIS increases the LoS rate by only about $0.028$~bps/Hz, whereas E-FAS provides an increase of approximately $14.3$~bps/Hz. This contrast is made explicit in Fig.~\ref{fig:ris}(c): the passive RIS gain over direct LoS rapidly approaches zero with distance, while the E-FAS gain increases to approximately $14.5$~bps/Hz around $D=76$~m and remains close to this level throughout the long-distance regime. These results demonstrate a fundamental propagation advantage of E-FAS over passive reflection: rather than relying on a two-hop space wave cascade, E-FAS transports the signal over most of the end-to-end distance through a substantially lower-loss guided pathway.}

{\color{blue}Fig.~\ref{fig:sr_m} evaluates the ZF sum-rate for $M=16$ and $K=4$ under the adopted trapezium geometry, where $d_{\rm DL}=D-1$ for the direct NLoS link and $d_{\rm sw}=D-2$ for the E-FAS route. The analytical results from~\eqref{eq-zf} closely agree with the Monte Carlo simulations over the entire distance and transmit power ranges, validating the ZF characterization under the  E-FAS channel. The results demonstrate a pronounced multiuser performance advantage of E-FAS at all considered transmit powers. While the direct NLoS sum-rate decreases rapidly with distance and eventually approaches zero, E-FAS maintains a high spatial multiplexing rate because the dominant propagation distance is transferred to the substantially lower loss guided-SW pathway. Moreover, increasing the BS transmit power translates almost directly into a higher E-FAS sum-rate over the entire distance range, whereas the NLoS link remains strongly distance limited. Consequently, the propagation advantage observed in the single-user analysis is preserved under multiuser transmission and enables E-FAS to sustain the ZF gain over distances for which conventional NLoS MU-MIMO becomes ineffective.}


\begin{figure}[]
\centering
\includegraphics[width=0.95\columnwidth]{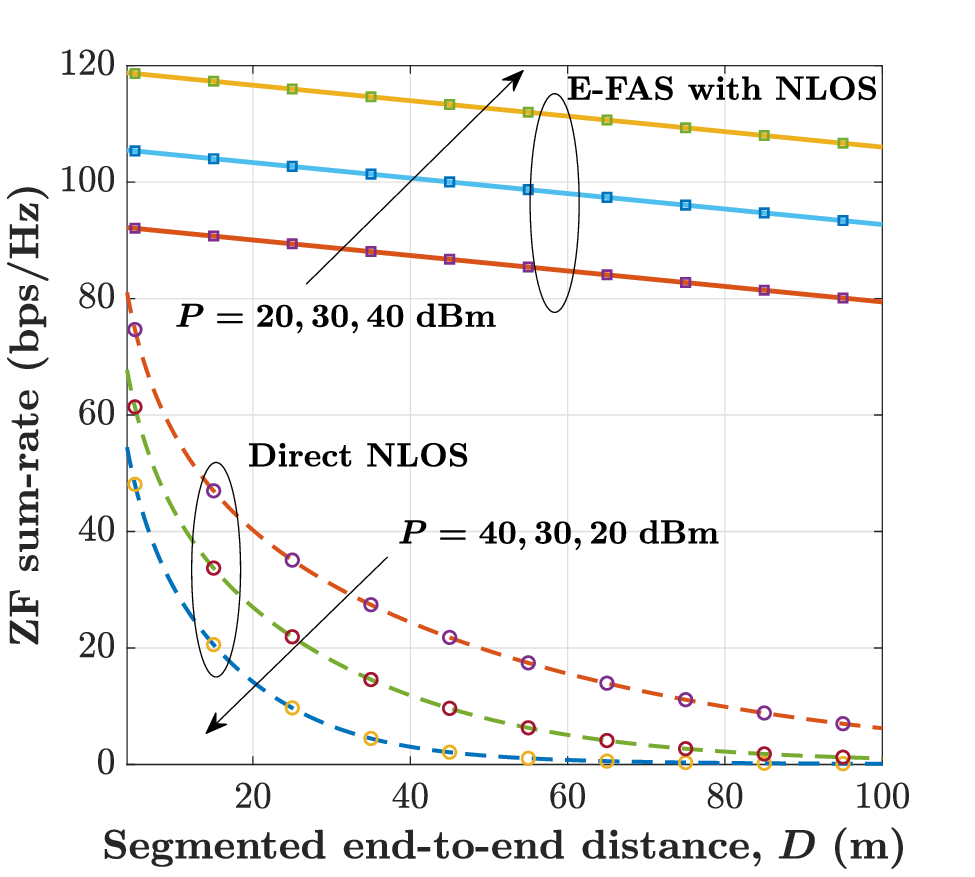}
\caption{{\color{blue}Multiuser ZF sum-rate versus the end-to-end distance under the segmented E-FAS channel model at $26.4$~GHz for $M=16$ and $K=4$. Lines denote the analytical result in~\eqref{eq-zf}, whereas markers represent Monte Carlo simulations of the multiuser Gaussian equivalent channel.}}\label{fig:sr_m}
\vspace{-2mm}
\end{figure}

\vspace{-2mm}
\section{Conclusion}\label{sec:conclusion}
{\color{blue}
We developed an end-to-end statistical channel model for E-FAS-assisted MU-MIMO transmission that integrates BS-to-surface excitation, guided-SW propagation, surface radiation, and the final space-wave segment. The SW-assisted channel was characterized as conditionally complex Gaussian with a random covariance scale, while its unconditional distribution is generally a Gaussian scale mixture. The framework incorporates spatial correlation, route-induced inter-user dependence, and practical SW-route perturbations, and identifies the channel hardening conditions under which the conventional spatially white Gaussian model becomes accurate. Under this tractable specialization, established Rayleigh fading and MU-MIMO results were applied to characterize single-user outage and ergodic performance and multiuser ZF transmission. Full cascade simulation results validated the proposed statistical characterization and confirmed the detrimental effect of spatial correlation across the antenna array at the BS on the available spatial DoF. Using experimentally reported E-FAS parameters at $26.4$~GHz, the numerical results demonstrated substantial reliability and rate gains over conventional NLoS transmission. The comparison with an ideal passive RIS further illustrated that E-FAS can provide a pronounced propagation advantage over moderate and long distances by transferring most of the end-to-end propagation distance to the lower loss guided-SW pathway. Consequently, the proposed framework provides a physically grounded and statistically consistent basis for analyzing E-FAS-assisted wireless systems.
}

\vspace{-2mm}
\bibliographystyle{IEEEtran}

\end{document}